\documentclass[sts,preprint]{imsart}
\makeatletter
\@twocolumnfalse
\makeatother
\RequirePackage{amsthm,amsmath,amsfonts,amssymb,mathtools}
\RequirePackage[authoryear]{natbib}
\RequirePackage[hyphens]{url}
\RequirePackage[colorlinks,citecolor=blue,urlcolor=blue]{hyperref}
\RequirePackage{graphicx}
\RequirePackage{booktabs}
\RequirePackage{tikz}

\startlocaldefs
\theoremstyle{plain}
\newtheorem{theorem}{Theorem}
\newtheorem{proposition}{Proposition}
\newtheorem{lemma}{Lemma}
\newtheorem{corollary}{Corollary}
\theoremstyle{remark}

\newtheorem{algorithm}{Algorithm}
\theoremstyle{definition}
\newtheorem{definition}{Definition}

\newcounter{labelledlistindex}
\newcounter{labelleditem}
\newenvironment{labelledlist}[1]%
  {\if@inlabel\leavevmode\fi
   \stepcounter{labelledlistindex}%
   \renewcommand{\thelabelleditem}{#1}%
   \begin{list}{\thelabelleditem}{\usecounter{labelleditem}}}%
  {\end{list}}
\newenvironment{algsteps}%
  {\begin{labelledlist}{(\alph{labelleditem})}%
   \let\algsteps@label\label
   \renewcommand{\label}[1]{\algsteps@label{##1}%
     {\protected@edef\@currentlabel{\thealgorithm.\alph{labelleditem}}\algsteps@label{##1:full}}}}%
  {\end{labelledlist}}
\newcommand{\algstep}[1]{Algorithm~\ref{#1:full}}

\newcommand{\dd}{\mathrm{d}}
\newcommand{\bbE}{\mathbb{E}}
\newcommand{\bbP}{\mathbb{P}}
\newcommand{\pos}{\mathbf{x}} %
\newcommand{\statespace}{\mathsf{X}} %
\newcommand{\target}{\pi} %
\newcommand{\exttarget}{\bar{\target}} %
\newcommand{\normconst}{Z} %
\newcommand{\transdens}{m} %
\newcommand{\pot}{g} %
\newcommand{\obsdens}{h} %
\newcommand{\weight}{W} %
\newcommand{\weightvec}{\mathbf{\weight}} %
\newcommand{\anc}{a} %
\newcommand{\refidx}{b} %
\newcommand{\outidx}{l} %
\newcommand{\bsprob}{B} %
\newcommand{\kernel}{K} %
\newcommand{\forced}{F} %
\newcommand{\coupling}{C} %
\newcommand{\expprod}{Q} %
\newcommand{\excess}{D} %
\newcommand{\dobrushin}{\beta} %
\newcommand{\mixconst}{\varepsilon} %
\newcommand{\gain}{\gamma} %
\newcommand{\resdist}{r} %
\newcommand{\AT}{\textsf{AT}}
\newcommand{\BS}{\textsf{BS}}
\DeclareRobustCommand{\FMAT}{\textsf{FM\nobreakdash-AT}}
\DeclareRobustCommand{\FMBS}{\textsf{FM\nobreakdash-BS}}
\DeclareRobustCommand{\MBS}{\textsf{M\nobreakdash-BS}}
\DeclareRobustCommand{\CMBS}{\textsf{CM\nobreakdash-BS}}
\DeclareRobustCommand{\legFMBSNtwo}{\tikz[baseline=-0.6ex]{\draw[color={rgb,255:red,213;green,94;blue,0}, line width=1pt, solid] (0,0) -- (16pt,0);}}
\DeclareRobustCommand{\legFMBSNeight}{\tikz[baseline=-0.6ex]{\draw[color={rgb,255:red,213;green,94;blue,0}, line width=1pt, densely dotted] (0,0) -- (16pt,0);}}
\DeclareRobustCommand{\legMBSNtwo}{\tikz[baseline=-0.6ex]{\draw[color={rgb,255:red,0;green,158;blue,115}, line width=1pt, solid] (0,0) -- (16pt,0);}}
\DeclareRobustCommand{\legMBSNeight}{\tikz[baseline=-0.6ex]{\draw[color={rgb,255:red,0;green,158;blue,115}, line width=1pt, densely dotted] (0,0) -- (16pt,0);}}
\DeclareRobustCommand{\legCMBSNtwo}{\tikz[baseline=-0.6ex]{\draw[color={rgb,255:red,0;green,114;blue,178}, line width=1pt, solid] (0,0) -- (16pt,0);}}
\DeclareRobustCommand{\legCMBSNeight}{\tikz[baseline=-0.6ex]{\draw[color={rgb,255:red,0;green,114;blue,178}, line width=1pt, densely dotted] (0,0) -- (16pt,0);}}
\endlocaldefs

\begin{document}

\begin{frontmatter}
\title{Metropolized backward sampling for conditional sequential Monte Carlo}
\runtitle{Metropolized backward sampling}

\begin{aug}
\author[A]{\fnms{Adrien}~\snm{Corenflos}\ead[label=e1]{adrien.corenflos@warwick.ac.uk}}
\address[A]{Department of Statistics, University of Warwick\printead[presep={\ }]{e1}.}
\end{aug}

\begin{abstract}
Particle Gibbs samplers with backward sampling are widely used for sampling the smoothing distribution of state-space models, but, with few particles, the new trajectory often coincides with the current one at many time steps.
A known but rarely used remedy replaces the random selection of the final particle by a Metropolis--Hastings move which avoids the current particle.
We show that, with backward sampling, this remedy only helps at the last time steps, because its effect on the update rates decays geometrically with the distance to the final time.
We propose instead Metropolized backward sampling, which makes such a move at every step of the backward pass, forcibly switching trajectories whenever possible.
Additionally, we introduce a coupled version which prevents proposing the current trajectory, maximizing the chances of an actual switch versus the current trajectory and performing perfect sampling when the state-space model has no observations, a feature not shared by backward sampling.
We prove that, at every time step, both methods dominate backward sampling as they update the state more often and have smaller lag-one autocorrelations than it does.
Our methods are drop-in replacements for backward sampling, and can be used to improve any particle Gibbs sampler.
\end{abstract}

\begin{keyword}
\kwd{Backward sampling}
\kwd{Markov chain Monte Carlo}
\kwd{Metropolized Gibbs sampler}
\kwd{particle Gibbs}
\kwd{Peskun ordering}
\kwd{state-space model}
\end{keyword}
\end{frontmatter}

\section{Introduction}\label{sec:intro}

\subsection{Particle Gibbs and backward sampling}

Particle Gibbs samplers~\citep{andrieu2010particle} sample the smoothing distribution of a state-space model, or, more generally, the path measure of a Feynman--Kac model~\citep{delmoral2004feynman, chopin2020introduction}, by iterating a \emph{conditional sequential Monte Carlo} (CSMC) update: a particle filter is run with one particle fixed to the current trajectory, the \emph{reference}, after which a new trajectory is selected among the particles.
In the original algorithm, the new trajectory is obtained by drawing a final-time particle and tracing its ancestry back in time, which we refer to as ancestral tracing~(\AT{}).
Because of path degeneracy, the early states of the reference are then rarely updated when the time horizon $T$ is large compared with the number of particles $N$.
Theoretical guarantees are then only obtained when $N$ grows linearly with $T$~\citep{andrieu2018uniform}.
\emph{Backward sampling}~(\BS{})~\citep{whiteley2010particle, lindsten2013backward}, which adapts the backward simulation of \citet{godsill2004monte} to CSMC, redraws the indices of the new trajectory backwards in time, from probabilities which combine the filtering weights with the transition densities.
It dominates \AT{} in asymptotic efficiency~\citep[Theorem~11]{chopin2015particle} and, contrary to \AT{}, remains effective with a number of particles which does not grow with $T$~\citep{lee2020coupled,karjalainen2026mixing}.

\subsection{Metropolized Gibbs steps and the forced move}

In CSMC, the final index, and, with \BS{}, each earlier index, is drawn by a Gibbs step on the extended space of all the variables generated by the particle filter~\citep{andrieu2010particle, chopin2015particle}, that is, from its conditional distribution, without regard for its current value.
However, for discrete variables, \citet{liu1996peskun} proposed the \emph{Metropolized Gibbs} sampler, which proposes a value different from the current one, with probabilities proportional to the conditional distribution, and accepts it with a Metropolis--Hastings probability.
It dominates the Gibbs step in the ordering of \citet{peskun1973optimum}, and hence in asymptotic variance~\citep{tierney1998note, mira2001ordering}.
With two candidates, the Gibbs step is the acceptance rule of \citet{barker1965monte}, whereas the Metropolized Gibbs step is the Metropolis rule.
\citet[Section~8]{chopin2013particle} applied this construction to the final index of CSMC with \AT{}, under the name of \emph{forced move}, and showed that the resulting sampler, which we denote by \FMAT{}~(forced move plus ancestral tracing), dominates \AT{} in the Peskun ordering~\citep{peskun1973optimum,tierney1998note}.
The construction extends immediately to \BS{}, which gives \FMBS{}~(forced move plus backward sampling).
The forced move does not appear in the published version~\citep{chopin2015particle}, which concentrates on \BS{}.
To our knowledge, it is rarely used in practice.
There is a simple reason for this: with \BS{}, a forced move at the final time step only changes the final index, while the backward draws at earlier times are exact draws from conditional distributions which forget the final index.
We show that, under mixing conditions, the benefit decays geometrically with the distance to the final time (Proposition~\ref{prop:washout}), so that, for even moderately long time series, \FMBS{} behaves like \BS{}.
This is particularly clear when the model factorizes across time, as, in this case, the backward draws are independent of the final index and a forced move at the final time step does not affect any earlier time steps.

\subsection{Contributions}

To address this issue, we propose to make a forced move at every step of the backward pass.
We call the resulting sampler \emph{Metropolized backward sampling}~(\MBS{}).
The key observation is that the index drawn at time $t$ has a natural current value: the ancestor of the particle selected at time $t + 1$, which \AT{} keeps deterministically, and which \BS{} ignores.
When the output and the reference differ at time $t + 1$, this ancestor is independent of the index of the reference at time $t$, which may therefore be selected again, decreasing efficiency.
To fix this, we introduce a \emph{coupled Metropolized backward sampling}~(\CMBS{}) variant.
The index at time $t$ is then drawn as with \BS{}, but replaced by a forced move away from the index of the reference if it equals the latter.
We prove that CSMC remains valid when each index of the backward pass is drawn from a Markov kernel \emph{started from the index of the reference}, under a symmetry condition which is directly satisfied by \AT{} and \BS{}, and which, when the output and the reference coincide at the next time step, is the reversibility with respect to the backward-sampling probabilities.
When the output and the reference differ, the kernels of \MBS{} draw the index independently of the index of the reference, so that they may be started from the ancestor instead.
This generalized backward pass is summarized in Algorithm~\ref{alg:csmc}.
It is reversible with respect to the target distribution (Proposition~\ref{prop:reversible}), which extends the reversibility of CSMC with \BS{}~\citep[Proposition~9]{chopin2015particle}.
At every time step, CSMC with \MBS{} or \CMBS{} has larger update rates and smaller lag-one autocovariances of all univariate functions of the state than \BS{} and \FMBS{} (Theorem~\ref{thm:ordering}).
The update rates of the states at two consecutive time steps are ordered in the same way, as are, for \MBS{}, their autocovariances (Corollary~\ref{cor:pairs}).
Given the indices of the reference and of the output at time $t + 1$, the probability that the index of the reference is selected again at time $t$ can only decrease when the draw of \FMBS{} at time $t$ is replaced by that of \MBS{} or \CMBS{}.
At every time $t < T$, the update rate of \MBS{}, or of \CMBS{}, exceeds that of \FMBS{} by at least the expectation of this decrease under that rule (Proposition~\ref{prop:gain}).
When the model factorizes across time and the particles have equal weights, every state is updated under \CMBS{}, whereas each state is kept with probability $1/N$ under \BS{} and, under \MBS{}, with a probability which tends to $1/(N + 1)$ as $T - t$ grows.
With two particles and a model which factorizes across time, CSMC with \CMBS{} performs an independent Metropolis--Hastings step at each time step (Proposition~\ref{prop:imh}).
Appendix~\ref{sec:previous-state} extends Algorithm~\ref{alg:csmc}, its reversibility, and the results of Sections~\ref{sec:update-rates} and~\ref{sec:gains} to potential functions which depend on two states and to any resampling scheme under which the expected number of offspring of each particle is $N$ times its normalized weight, with ancestor sampling of the reference.
Section~\ref{sec:numerics} illustrates these results on a linear-Gaussian model and on the nonlinear model of \citet{gordon1993novel}, which \citet{andrieu2010particle} used to illustrate CSMC.
The proofs of our results are given in Appendix~\ref{sec:proofs}.

\subsection{Related work}

In particle Gibbs with ancestor sampling, \citet[Section~6.1]{lindsten2014particle} observed that the ancestor index of the reference may be drawn from any Markov kernel which leaves the ancestor-sampling probabilities invariant.
The backward pass of Algorithm~\ref{alg:csmc} applies this idea to backward sampling, with kernels started from the indices of the reference.
To reduce the cost of ancestor sampling, \citet{lindsten2014particle} suggested a Metropolis--Hastings step, which only evaluates the ancestor-sampling weights of the current and proposed indices.
They recommended a proposal which never proposes the current index, in the spirit of \citet{liu1996peskun}, without specifying it further.
Such a kernel is not the Metropolized Gibbs kernel of \citet{liu1996peskun}, whose proposal requires all the weights, so that, in general, it does not dominate the Gibbs step in the Peskun ordering.
In particle smoothing, \citet{bunch2013improved} also replaced each backward draw by Metropolis--Hastings steps, started from the ancestor of the particle selected at the next time step, to reduce the cost of backward sampling.
\citet{dau2023backward} formulated these steps as independent Metropolis--Hastings steps with the filtering weights as proposal.
They proved that a single step per draw gives a convergent and stable smoother, at a cost linear in $N$ \citep[Proposition~4]{dau2023backward}.
Their validity rests on the same invariance as that of the forced moves of \MBS{} started from the ancestors: the ancestor is distributed according to the backward-sampling probabilities, which these kernels leave invariant (\citealp[Theorem~2.1]{dau2023backward}; Lemma~\ref{lem:pair}).
We use the Metropolized Gibbs kernel instead, whose cost is that of \BS{}, in order to improve the statistical efficiency of the sampler rather than its cost.
In the backward pass of CSMC, the steps of \citet{bunch2013improved} would have the opposite effect of \MBS{}: because they favour the current index, they would reduce the update rates and increase the autocorrelations.

\section{Backward sampling and forced moves}\label{sec:method}

\subsection{Feynman--Kac models and conditional SMC}

Let $\statespace$ be a measurable space, equipped with a $\sigma$-finite measure $\dd x$.
We consider a Markov chain on $\statespace$ with initial density $\transdens_1(x_1)$ and transition densities $\transdens_t(x_{t-1}, x_t)$, $2 \leq t \leq T$, with respect to $\dd x$, and potential functions $\pot_t \colon \statespace \to (0, \infty)$, $t \in [T] \coloneqq \{1, \dotsc, T\}$.
The target distribution has density
\begin{equation}
  \target(x_{1:T}) = {\frac{\transdens_1(x_1) \pot_1(x_1)}{\normconst} \prod_{t=2}^T \transdens_t(x_{t-1}, x_t) \pot_t(x_t)}, \label{eq:target}
\end{equation}
where we assume that $\normconst \in (0, \infty)$.
For a state-space model with transition densities $\transdens_t$ and observation densities $\obsdens_t(y_t \mid x_t)$, the choice $\pot_t(x_t) = \obsdens_t(y_t \mid x_t)$ makes $\target$ the smoothing distribution and the particle filter below the bootstrap filter.
The particle filter uses $N \geq 2$ particles $x_t^1, \dotsc, x_t^N$ at each time step, with the normalized weights $\weight_t^n \coloneqq \pot_t(x_t^n) / \sum_{i=1}^N \pot_t(x_t^i)$, which form the probability vector $\weightvec_t \coloneqq (\weight_t^1, \dotsc, \weight_t^N)$.
We denote by $\anc_{t-1}^n \in [N]$ the index of the ancestor at time $t - 1$ of the particle $x_t^n$.
For $t \in [T-1]$ and $i, j \in [N]$, we define the \emph{backward-sampling probabilities}
\begin{equation}
  \bsprob_t(i \mid j) \coloneqq \frac{\weight_t^i \transdens_{t+1}(x_t^i, x_{t+1}^j)}{\sum_{k=1}^N \weight_t^k \transdens_{t+1}(x_t^k, x_{t+1}^j)}. \label{eq:bs}
\end{equation}

\begin{algorithm}[CSMC with a general backward pass]\label{alg:csmc}
Given the reference trajectory $x_{1:T}$ and Markov kernels on $[N]$ which are functions of the particle positions only: $\kernel_T$, reversible with respect to $\weightvec_T$, and, for $t \in [T-1]$ and $j, j' \in [N]$, $\kernel_{t,j,j'}$, such that, for all $i, i' \in [N]$,
\begin{equation}
  \bsprob_t(i \mid j) \kernel_{t,j,j'}(i, i') = \bsprob_t(i' \mid j') \kernel_{t,j',j}(i', i): \label{eq:symmetry}
\end{equation}
\begin{algsteps}
  \item\label{alg:csmc:indices} Draw $\refidx_1, \dotsc, \refidx_T$ independent uniforms on $[N]$.
  \item\label{alg:csmc:forward} Set $x_1^{\refidx_1} \coloneqq x_1$ and draw $x_1^n \sim \transdens_1$ for $n \neq \refidx_1$.
For $t = 2, \dotsc, T$, set $\anc_{t-1}^{\refidx_t} \coloneqq \refidx_{t-1}$ and $x_t^{\refidx_t} \coloneqq x_t$, then, for $n \neq \refidx_t$, draw $\anc_{t-1}^n \sim \weightvec_{t-1}$ and $x_t^n \sim \transdens_t(x_{t-1}^{\anc_{t-1}^n}, \cdot)$.
  \item\label{alg:csmc:final} Draw $\outidx_T \sim \kernel_T(\refidx_T, \cdot)$.
  \item\label{alg:csmc:backward} For $t = T - 1, \dotsc, 1$, draw $\outidx_t \sim \kernel_{t, \refidx_{t+1}, \outidx_{t+1}}(\refidx_t, \cdot)$.
  \item\label{alg:csmc:return} Return $\tilde{x}_{1:T} \coloneqq (x_1^{\outidx_1}, \dotsc, x_T^{\outidx_T})$.
\end{algsteps}
\end{algorithm}

The kernels $\kernel_T$ and $\kernel_{t,j,j'}$ of the six selection rules below depend on the particles only through the probabilities $\weightvec_T$ and $\bsprob_t$.
In \algstep{alg:csmc:backward}, the kernel at time $t$ starts from the index $\refidx_t$ of the reference.
If $\outidx_{t+1} = \refidx_{t+1}$, this is the ancestor $\anc_t^{\outidx_{t+1}}$ of the particle selected at time $t + 1$.
For $j = j'$, \eqref{eq:symmetry} is the reversibility of $\kernel_{t,j,j}$ with respect to $\bsprob_t(\cdot \mid j)$.
For $j \neq j'$, summing~\eqref{eq:symmetry} over $i$ shows that, if $i \sim \bsprob_t(\cdot \mid j)$ and $i' \sim \kernel_{t,j,j'}(i, \cdot)$, then $i' \sim \bsprob_t(\cdot \mid j')$, while~\eqref{eq:symmetry} states that the pair $(i', i)$ then has the same distribution as when $i' \sim \bsprob_t(\cdot \mid j')$ and $i \sim \kernel_{t,j',j}(i', \cdot)$.
For probability vectors $\mathbf{p}$ and $\mathbf{q}$ on $[N]$, the selection rules are built from the forced move $\forced_{\mathbf{p}}$ and the coupling kernel $\coupling_{\mathbf{p},\mathbf{q}}$ of $\mathbf{p}$ and $\mathbf{q}$, whose formal definitions are deferred to Sections~\ref{sec:forced} and~\ref{sec:coupling}.
We consider six kernel choices.
In the first five, $\kernel_{t,j,j'}(i, \cdot) = \bsprob_t(\cdot \mid j')$ for $j \neq j'$, so that the indices of the output and of the reference are drawn independently once they differ:
\begin{list}{}{\settowidth{\labelwidth}{\CMBS{}:}\setlength{\leftmargin}{\dimexpr\labelwidth+\labelsep\relax}}
  \item[\AT{}:] $\kernel_T(i, \cdot) = \weightvec_T$ and $\kernel_{t,j,j}(i, \cdot) = \delta_i$: ancestral tracing~\citep{andrieu2010particle};
  \item[\FMAT{}:] $\kernel_T = \forced_{\weightvec_T}$ and $\kernel_{t,j,j}(i, \cdot) = \delta_i$: forced move plus ancestral tracing~\citep[Section~8]{chopin2013particle};
  \item[\BS{}:] $\kernel_T(i, \cdot) = \weightvec_T$ and $\kernel_{t,j,j}(i, \cdot) = \bsprob_t(\cdot \mid j)$: the backward sampling of \citet{whiteley2010particle};
  \item[\FMBS{}:] $\kernel_T = \forced_{\weightvec_T}$ and $\kernel_{t,j,j}(i, \cdot) = \bsprob_t(\cdot \mid j)$: forced move plus backward sampling;
  \item[\MBS{}:] $\kernel_T = \forced_{\weightvec_T}$ and $\kernel_{t,j,j} = \forced_{\bsprob_t(\cdot \mid j)}$: Metropolized backward sampling, which we propose;
  \item[\CMBS{}:] $\kernel_T = \forced_{\weightvec_T}$, $\kernel_{t,j,j} = \forced_{\bsprob_t(\cdot \mid j)}$, and, for $j \neq j'$, $\kernel_{t,j,j'} = \coupling_{\bsprob_t(\cdot \mid j), \bsprob_t(\cdot \mid j')}$: coupled Metropolized backward sampling, which we also propose.
\end{list}
For the first five rules, \algstep{alg:csmc:backward} is usually implemented as $\outidx_t \sim \kernel_{t, \outidx_{t+1}, \outidx_{t+1}}(\anc_t^{\outidx_{t+1}}, \cdot)$, from the ancestor of the particle selected at time $t + 1$, without the indices of the reference at the times $t < T$.
When $x_{1:T} \sim \target$, this gives the same joint distribution of the reference and the output (Lemma~\ref{lem:pair}).
With \AT{}, $\outidx_t = \anc_t^{\outidx_{t+1}}$, whereas, with \MBS{}, $\outidx_t$ is drawn by a forced move from $\anc_t^{\outidx_{t+1}}$ at every time step.
With \CMBS{}, if $\outidx_{t+1} = \refidx_{t+1}$, the index $\outidx_t$ is drawn as with \MBS{}, by a forced move from the ancestor $\refidx_t = \anc_t^{\outidx_{t+1}}$.
Otherwise, it is drawn from $\bsprob_t(\cdot \mid \outidx_{t+1})$, as with \BS{}, and replaced, if it equals $\refidx_t$, by a forced move from $\refidx_t$ with respect to the probabilities proportional to $\bsprob_t(i \mid \refidx_{t+1}) \bsprob_t(i \mid \outidx_{t+1})$, $i \in [N]$.
As for \BS{}~\citep[Remark~2]{chopin2015particle}, the law of the output does not change if the indices of the reference are fixed, $\refidx_1 = \dotsb = \refidx_T = 1$, because the steps of the particle filter are exchangeable and the six selection rules are equivariant under relabelling of the particles.
We use uniformly drawn indices in the proofs and fixed indices in the experiments.
All six rules require the probabilities $\weightvec_T$, and, except for \AT{} and \FMAT{}, the probabilities $\bsprob_t(\cdot \mid \outidx_{t+1})$, $t \in [T-1]$.
\CMBS{} also requires the probabilities $\bsprob_t(\cdot \mid \refidx_{t+1})$ at the times $t$ at which $\outidx_{t+1} \neq \refidx_{t+1}$ and the draw from $\bsprob_t(\cdot \mid \outidx_{t+1})$ equals $\refidx_t$.
The costs of \BS{}, \FMBS{}, \MBS{}, and \CMBS{} are therefore of order $NT$.
The backward pass of \MBS{} costs exactly as much as that of \BS{}, and that of \CMBS{} at most twice as much.

\subsection{The forced move}\label{sec:forced}

Let $\mathbf{p}$ be a probability vector on $[N]$.
Under the Gibbs kernel, $\kernel(i, \cdot) = \mathbf{p}$, the new index is drawn from $\mathbf{p}$ without regard for the current index $i$.
The Metropolized Gibbs kernel of \citet{liu1996peskun}, which is the forced move of \citet{chopin2013particle}, proposes, from $i \in [N]$ with $p_i < 1$, an index $j \neq i$ with probability $p_j / (1 - p_i)$, and accepts it with probability $\min\{1, (1 - p_i) / (1 - p_j)\}$, so that
\begin{equation}
\begin{split}
  \forced_{\mathbf{p}}(i, j) & = \frac{p_j}{\max\{1 - p_i, 1 - p_j\}}, \quad j \neq i, \\
  \forced_{\mathbf{p}}(i, i) & = 1 - \sum_{j \neq i} \forced_{\mathbf{p}}(i, j),
\end{split} \label{eq:forced}
\end{equation}
and $\forced_{\mathbf{p}}(i, \cdot) \coloneqq \delta_i$ if $p_i = 1$.
For $N = 2$ and $p_1, p_2 > 0$, $\forced_{\mathbf{p}}(1, 2) = \min\{1, p_2 / p_1\}$ is the Metropolis acceptance probability, whereas the probability $p_2 = p_2 / (p_1 + p_2)$ of moving from $1$ to $2$ under the Gibbs kernel is that of \citet{barker1965monte}.
For a Markov kernel $\kernel$ on $[N]$ and $\varphi \colon [N] \to \mathbb{R}$, we write $\mathbf{p}(\varphi) \coloneqq \sum_i p_i \varphi(i)$ for the expectation of $\varphi$ with respect to $\mathbf{p}$, and
\begin{align*}
  \langle \varphi, \kernel \varphi \rangle_{\mathbf{p}} \coloneqq \sum_{i, j \in [N]} p_i \kernel(i, j) \varphi(i) \varphi(j).
\end{align*}
This is the expectation of $\varphi(i) \varphi(j)$ when $i \sim \mathbf{p}$ and $j \sim \kernel(i, \cdot)$, so that, if $\kernel$ leaves $\mathbf{p}$ invariant, $\langle \varphi, \kernel \varphi \rangle_{\mathbf{p}} - \mathbf{p}(\varphi)^2$ is the lag-one autocovariance of $\varphi$ for the Markov chain with kernel $\kernel$ started from $\mathbf{p}$.

\begin{lemma}\label{lem:forced}
The kernel $\forced_{\mathbf{p}}$ is reversible with respect to $\mathbf{p}$.
For $i \neq j$, $\forced_{\mathbf{p}}(i, j) \geq p_j$.
For every $\varphi \colon [N] \to \mathbb{R}$, $\langle \varphi, \forced_{\mathbf{p}} \varphi \rangle_{\mathbf{p}} \leq \mathbf{p}(\varphi)^2$.
\end{lemma}

The first two parts of Lemma~\ref{lem:forced} are due to \citet[Section~2]{liu1996peskun} when $p_i < 1$ for all $i \in [N]$.
Because the probability of moving from $i$ to $j \neq i$ under the Gibbs kernel is $p_j$, the second part states that $\forced_{\mathbf{p}}$ dominates the Gibbs kernel off the diagonal, which is the ordering of \citet{peskun1973optimum}.
If $p_i = 1$ for some $i$, then the first two parts follow from $p_j = 0$ and $\forced_{\mathbf{p}}(j, i) = 1$ for $j \neq i$.
Because $\langle \varphi, \kernel \varphi \rangle_{\mathbf{p}} = \mathbf{p}(\varphi)^2$ for the Gibbs kernel, the third part follows from the first two by \citet[Lemma~3]{tierney1998note}.

\subsection{A coupling which avoids the current index}\label{sec:coupling}

The forced move also gives a coupling of two probability vectors $\mathbf{p}$ and $\mathbf{q}$ on $[N]$ under which the two indices coincide less often than when they are independent: this is the opposite of classically used maximal couplings~\citep{goldstein1979maximal}.
For $\mathbf{u} \in [0, \infty)^N$, let $\forced_{\mathbf{u}} \coloneqq \forced_{\mathbf{u} / (u_1 + \dotsb + u_N)}$ if $\mathbf{u} \neq \mathbf{0}$ and $\forced_{\mathbf{0}}(i, \cdot) \coloneqq \delta_i$.
With $\mathbf{p} \odot \mathbf{q} \coloneqq (p_1 q_1, \dotsc, p_N q_N)$, let
\begin{equation}
  \coupling_{\mathbf{p},\mathbf{q}}(i, \cdot) \coloneqq \mathbf{q} + q_i \{\forced_{\mathbf{p} \odot \mathbf{q}}(i, \cdot) - \delta_i\}, \quad i \in [N]. \label{eq:coupling}
\end{equation}
That is, an index is drawn from $\mathbf{q}$ and replaced, if it equals $i$, by a forced move from $i$ with respect to the probabilities proportional to $p_j q_j$, $j \in [N]$.
In particular, $\coupling_{\mathbf{p},\mathbf{q}}(i, i) = q_i \forced_{\mathbf{p} \odot \mathbf{q}}(i, i)$, so that the entries of $\coupling_{\mathbf{p},\mathbf{q}}$ are nonnegative and $\coupling_{\mathbf{p},\mathbf{q}}$ is a Markov kernel.

\begin{lemma}\label{lem:coupling}
Let $\mathbf{p}$ and $\mathbf{q}$ be probability vectors on $[N]$.
For all $i, i' \in [N]$, $p_i \coupling_{\mathbf{p},\mathbf{q}}(i, i') = q_{i'} \coupling_{\mathbf{q},\mathbf{p}}(i', i)$.
For $i \neq i'$, $\coupling_{\mathbf{p},\mathbf{q}}(i, i') \geq q_{i'}$.
For every $\varphi \colon [N] \to \mathbb{R}$, $\langle \varphi, \coupling_{\mathbf{p},\mathbf{q}} \varphi \rangle_{\mathbf{p}} \leq \mathbf{p}(\varphi) \mathbf{q}(\varphi)$.
\end{lemma}

By the first part of Lemma~\ref{lem:coupling}, if $i \sim \mathbf{p}$ and $i' \sim \coupling_{\mathbf{p},\mathbf{q}}(i, \cdot)$, then $i' \sim \mathbf{q}$.
The two indices then coincide with probability $\sum_i p_i q_i \forced_{\mathbf{p} \odot \mathbf{q}}(i, i)$, whereas independent indices coincide with probability $\sum_i p_i q_i$.
If $\mathbf{p} = \mathbf{q}$ is uniform, then $\forced_{\mathbf{p} \odot \mathbf{q}}(i, i) = 0$ for all $i \in [N]$, so that the two indices never coincide.
The second and third parts of Lemma~\ref{lem:coupling} compare $\coupling_{\mathbf{p},\mathbf{q}}$ with the independent coupling, whose rows all equal $\mathbf{q}$, as the second and third parts of Lemma~\ref{lem:forced} compare $\forced_{\mathbf{p}}$ with the Gibbs kernel.

\section{Validity}\label{sec:validity}

Our proofs rely on the extended distribution of \citet{andrieu2010particle} on the particle positions $\pos \coloneqq (x_t^n)_{t \in [T], n \in [N]}$, the ancestor indices $\anc \coloneqq (\anc_t^n)_{t \in [T-1], n \in [N]}$, and a final index $\refidx_T \in [N]$, which have joint density
\begin{equation}
\begin{split}
  \exttarget(\pos, \anc, \refidx_T) \coloneqq {} & {\frac{\weight_T^{\refidx_T}}{\normconst N^T} \prod_{t=1}^T \Bigl\{\sum_{i=1}^N \pot_t(x_t^i)\Bigr\} \prod_{n=1}^N \transdens_1(x_1^n)} \\
  & \times \prod_{t=2}^T \prod_{n=1}^N \weight_{t-1}^{\anc_{t-1}^n} \transdens_t(x_{t-1}^{\anc_{t-1}^n}, x_t^n).
\end{split} \label{eq:extended}
\end{equation}
The trajectory traced back from $\refidx_T$ is $(x_1^{\refidx_1}, \dotsc, x_T^{\refidx_T})$, with genealogy $\refidx_t \coloneqq \anc_t^{\refidx_{t+1}}$ for $t \in [T-1]$.
For the variables generated by steps~\ref{alg:csmc:indices} and~\ref{alg:csmc:forward} of Algorithm~\ref{alg:csmc}, these indices are those of the reference drawn in \algstep{alg:csmc:indices}, because $\anc_t^{\refidx_{t+1}} = \refidx_t$ by \algstep{alg:csmc:forward}.

\begin{lemma}\label{lem:extended}
Let $x_{1:T} \sim \target$ and run steps~\ref{alg:csmc:indices} and~\ref{alg:csmc:forward} of Algorithm~\ref{alg:csmc}.
Then $(\pos, \anc, \refidx_T) \sim \exttarget$, with $x_{1:T}$ the trajectory traced back from $\refidx_T$.
Under $\exttarget$, conditionally on $\pos$, the indices $\refidx_T$ and $\anc_t^n$, $t \in [T-1]$, $n \in [N]$, are independent, with the distributions $\weightvec_T$ and $\bsprob_t(\cdot \mid n)$, respectively.
\end{lemma}

The first part of Lemma~\ref{lem:extended} is due to \citet{andrieu2010particle}, in the form of \citet[Section~2.3]{chopin2015particle}, where the indices of the reference are drawn uniformly before the other particles are generated.
The second part of Lemma~\ref{lem:extended} is the observation of \citet[proof of Proposition~9]{chopin2015particle} that, with multinomial resampling, \BS{} is a Gibbs update of all the ancestor indices given the particle positions.
It relies on the potential functions depending on the current state only, so that the weights do not depend on the ancestor indices.

Combined with steps~\ref{alg:csmc:final} and~\ref{alg:csmc:backward} of Algorithm~\ref{alg:csmc}, Lemma~\ref{lem:extended} gives the joint distribution of the indices $\refidx_{1:T}$ of the reference and $\outidx_{1:T}$ of the output trajectory: given the particle positions, the pairs $(\refidx_t, \outidx_t)$ form a Markov chain backwards in time.

\begin{lemma}\label{lem:pair}
Let $x_{1:T} \sim \target$ and run Algorithm~\ref{alg:csmc}.
Conditionally on $\pos$, $\refidx_T \sim \weightvec_T$ and, given $\refidx_T$, $\outidx_T \sim \kernel_T(\refidx_T, \cdot)$.
For $t \in [T-1]$, conditionally on $\pos$ and $(\refidx_{t+1:T}, \outidx_{t+1:T})$, $\refidx_t \sim \bsprob_t(\cdot \mid \refidx_{t+1})$ and, given $\refidx_t$, $\outidx_t \sim \kernel_{t, \refidx_{t+1}, \outidx_{t+1}}(\refidx_t, \cdot)$.
If $\kernel_{t,j,j'}(i, \cdot) = \bsprob_t(\cdot \mid j')$ for all $j \neq j'$ and $t \in [T-1]$, the same holds with $\outidx_t \sim \kernel_{t, \outidx_{t+1}, \outidx_{t+1}}(\anc_t^{\outidx_{t+1}}, \cdot)$ in \algstep{alg:csmc:backward}.
\end{lemma}

Special cases of Proposition~\ref{prop:reversible} are known: for \AT{} (\citealp[Proposition~8]{chopin2015particle}; \citealp[Theorem~1]{andrieu2018uniform}), for \BS{}~\citep[Proposition~9]{chopin2015particle}, and for backward sampling at a subset of the time steps combined with ancestral tracing at the others~\citep[Proposition~15]{chopin2013particle}.
We are not aware of a proof for general kernels $\kernel_T$ and $\kernel_{t,j,j'}$, so we give one here.

\begin{proposition}\label{prop:reversible}
For any kernels $\kernel_T$ and $\kernel_{t,j,j'}$ as in Algorithm~\ref{alg:csmc}, the Markov kernel defined by Algorithm~\ref{alg:csmc} is reversible with respect to $\target$.
In particular, it leaves $\target$ invariant.
\end{proposition}

By Lemma~\ref{lem:pair}, conditionally on the particle positions, the pairs $(\refidx_t, \outidx_t)$ form a backward Markov chain on $[N]^2$ whose initial distribution and transitions, by the reversibility of $\kernel_T$ and by~\eqref{eq:symmetry}, do not change when the indices of the reference and of the output are exchanged, which gives the reversibility.

Proposition~\ref{prop:reversible} applies to the six selection rules of Section~\ref{sec:method}, because the identity kernel, the Gibbs kernel, and, by Lemma~\ref{lem:forced}, $\forced_{\mathbf{p}}$ are reversible with respect to $\mathbf{p}$, whereas, for $j \neq j'$, the kernels $\kernel_{t,j,j'}(i, \cdot) = \bsprob_t(\cdot \mid j')$ satisfy~\eqref{eq:symmetry}, and so do the kernels $\kernel_{t,j,j'} = \coupling_{\bsprob_t(\cdot \mid j), \bsprob_t(\cdot \mid j')}$ of \CMBS{}, by the first part of Lemma~\ref{lem:coupling}.
For \BS{}, by Lemma~\ref{lem:extended}, $\outidx_{1:T}$ has the distribution of the indices traced back from an independent draw of $(\refidx_T, \anc)$ from its conditional distribution given $\pos$.
Proposition~\ref{prop:reversible} then follows from the reversibility of this Gibbs update, which is the proof of \citet[Proposition~9]{chopin2015particle}.

\section{Comparison of the selection rules}\label{sec:comparison}

\subsection{Update rates and autocovariances}\label{sec:update-rates}

Throughout this section, $x_{1:T} \sim \target$ and $\tilde{x}_{1:T}$ is the output of one iteration of Algorithm~\ref{alg:csmc}, so that $x_t = x_t^{\refidx_t}$ and $\tilde{x}_t = x_t^{\outidx_t}$ for $t \in [T]$.
Probabilities and expectations conditional on $\pos$ are denoted by $\bbP_{\pos}$ and $\bbE_{\pos}$, while those with the name of a selection rule as subscript, such as $\bbP_{\MBS{}}$, are taken under the joint distribution of all the variables, with $x_{1:T} \sim \target$ and $\tilde{x}_{1:T}$ the output of Algorithm~\ref{alg:csmc} for said selection rule.

The comparisons rely on a backward recursion for the expected products of a function at the indices of the reference and of the output.
For $\varphi \colon [N] \to \mathbb{R}$, $t \in [T]$, $s \in [T-1]$, and $j, j' \in [N]$, let
\begin{align*}
  \expprod_t(\varphi) & \coloneqq \bbE_{\pos}\{\varphi(\refidx_t) \varphi(\outidx_t)\}, \\
  (\bsprob_s \varphi)(j) & \coloneqq \sum_{i=1}^N \bsprob_s(i \mid j) \varphi(i), \\
  \excess_{s,j,j'}(\varphi) & \coloneqq \langle \varphi, \kernel_{s,j,j'} \varphi \rangle_{\bsprob_s(\cdot \mid j)} - (\bsprob_s \varphi)(j) (\bsprob_s \varphi)(j').
\end{align*}
By Lemma~\ref{lem:pair}, $\expprod_T(\varphi) = \langle \varphi, \kernel_T \varphi \rangle_{\weightvec_T}$.
For $t \in [T-1]$, also by Lemma~\ref{lem:pair}, $\langle \varphi, \kernel_{t,j,j'} \varphi \rangle_{\bsprob_t(\cdot \mid j)}$ is the conditional expectation of $\varphi(\refidx_t) \varphi(\outidx_t)$ given $\refidx_{t+1} = j$ and $\outidx_{t+1} = j'$, while it equals $(\bsprob_t \varphi)(j) (\bsprob_t \varphi)(j') + \excess_{t,j,j'}(\varphi)$ by definition, so that
\begin{equation}
\begin{split}
  \expprod_t(\varphi) = {} & \expprod_{t+1}(\bsprob_t \varphi) \\
  & + \sum_{j, j' \in [N]} \bbP_{\pos}(\refidx_{t+1} = j, \outidx_{t+1} = j') \excess_{t,j,j'}(\varphi).
\end{split} \label{eq:qt}
\end{equation}
For \BS{} and \FMBS{}, $\excess_{t,j,j'}(\varphi) = 0$.
For \AT{}, $\excess_{t,j,j'}(\varphi)$ is the variance of $\varphi$ under $\bsprob_t(\cdot \mid j)$ if $j = j'$ and zero otherwise, whereas, for \MBS{} and \CMBS{}, $\excess_{t,j,j'}(\varphi) \leq 0$, by Lemma~\ref{lem:forced} if $j = j'$, and, if $j \neq j'$, because $\excess_{t,j,j'}(\varphi) = 0$ for \MBS{} and by Lemma~\ref{lem:coupling} for \CMBS{}.
Backward induction on~\eqref{eq:qt} therefore orders the selection rules at every time step.

\begin{theorem}\label{thm:ordering}
Let $t \in [T]$ and let $f \colon \statespace \to \mathbb{R}$ be measurable with $\int f(x_t)^2 \target(x_{1:T}) \dd x_{1:T} < \infty$.
Then
\begin{align*}
  \bbE_{\MBS{}}\{f(x_t) f(\tilde{x}_t)\} & \leq \bbE_{\FMBS{}}\{f(x_t) f(\tilde{x}_t)\} \\
  & \leq \bbE_{\BS{}}\{f(x_t) f(\tilde{x}_t)\} \\
  & \leq \bbE_{\AT{}}\{f(x_t) f(\tilde{x}_t)\},
\end{align*}
and the first inequality also holds with \CMBS{} in place of \MBS{}.
Moreover,
\begin{align*}
  \bbP_{\MBS{}}(\tilde{x}_t \neq x_t) & \geq \bbP_{\FMBS{}}(\tilde{x}_t \neq x_t) \\
  & \geq \bbP_{\BS{}}(\tilde{x}_t \neq x_t) \geq \bbP_{\AT{}}(\tilde{x}_t \neq x_t),
\end{align*}
and the first inequality also holds with \CMBS{} in place of \MBS{}.
\end{theorem}

By stationarity, $f(x_t)$ and $f(\tilde{x}_t)$ have the same variance, so that the lag-one autocorrelations are ordered in the same way.
If $\statespace = \mathbb{R}^d$ and $\int \lVert x_t \rVert^2 \target(x_{1:T}) \dd x_{1:T} < \infty$, then the expected squared jumped distances (ESJD)
\begin{align*}
  \bbE\lVert \tilde{x}_t - x_t \rVert^2 = 2 \bbE\lVert x_t \rVert^2 - 2 \bbE\langle x_t, \tilde{x}_t \rangle
\end{align*}
are ordered in the reverse direction.
At $t = T$, the pair $(x_T, \tilde{x}_T)$ has the same distribution under \FMBS{} as under \FMAT{}, and under \BS{} as under \AT{}.
For bounded potential functions, the comparison of \FMBS{} and \BS{} at $t = T$ therefore also follows from the Peskun ordering of \FMAT{} over \AT{}~\citep[Proposition~7]{chopin2013particle}, combined with \citet[Lemma~3]{tierney1998note} for the autocovariances.
The comparison of \BS{} and \AT{} is a weaker form of \citet[Theorem~10]{chopin2015particle}, who prove lag-one domination for all square-integrable functions of the trajectory.

While Theorem~\ref{thm:ordering} applies to univariate functions, the same argument applies to functions of the states at two consecutive time steps, except for the autocovariances of \CMBS{}.

\begin{corollary}\label{cor:pairs}
Let $t \in [T-1]$ and let $f \colon \statespace^2 \to \mathbb{R}$ be measurable with $\int f(x_t, x_{t+1})^2 \target(x_{1:T}) \dd x_{1:T} < \infty$.
Then
\begin{align*}
  & \bbE_{\MBS{}}\{f(x_t, x_{t+1}) f(\tilde{x}_t, \tilde{x}_{t+1})\} \\
  & \quad \leq \bbE_{\FMBS{}}\{f(x_t, x_{t+1}) f(\tilde{x}_t, \tilde{x}_{t+1})\} \\
  & \quad \leq \bbE_{\BS{}}\{f(x_t, x_{t+1}) f(\tilde{x}_t, \tilde{x}_{t+1})\} \\
  & \quad \leq \bbE_{\AT{}}\{f(x_t, x_{t+1}) f(\tilde{x}_t, \tilde{x}_{t+1})\}.
\end{align*}
If, moreover, the particles at time $t$ and those at time $t + 1$ are almost surely pairwise distinct under $\exttarget$, then the probabilities of $\{(\tilde{x}_t, \tilde{x}_{t+1}) \neq (x_t, x_{t+1})\}$ are nonincreasing from \MBS{} to \FMBS{}, \BS{}, and \AT{}, and are at least as large under \CMBS{} as under \FMBS{}.
\end{corollary}

The particles are pairwise distinct, for example, if $\statespace = \mathbb{R}^d$ and $\dd x$ is Lebesgue measure.
With $\varphi(i, i') \coloneqq f(x_t^i, x_{t+1}^{i'})$ for $i, i' \in [N]$, by Lemma~\ref{lem:pair}, the conditional expectation of $f(x_t, x_{t+1}) f(\tilde{x}_t, \tilde{x}_{t+1})$ given $\pos$, $\refidx_{t+1} = j$, and $\outidx_{t+1} = j'$ is
\begin{equation}
  \sum_{i, i' \in [N]} \bsprob_t(i \mid j) \kernel_{t,j,j'}(i, i') \varphi(i, j) \varphi(i', j'). \label{eq:pairs}
\end{equation}
For \CMBS{} and $j \neq j'$, the functions $\varphi(\cdot, j)$ and $\varphi(\cdot, j')$ are applied in~\eqref{eq:pairs} to the two indices of $\coupling_{\mathbf{p},\mathbf{q}}$, where $\mathbf{p} = \bsprob_t(\cdot \mid j)$ and $\mathbf{q} = \bsprob_t(\cdot \mid j')$.
If $\varphi(\cdot, j') = \varphi(\cdot, j)$, Lemma~\ref{lem:coupling} bounds~\eqref{eq:pairs} by its value for the independent coupling, but the inequality is reversed if $\varphi(\cdot, j') = -\varphi(\cdot, j)$, so that Corollary~\ref{cor:pairs} does not compare the autocovariances of \CMBS{} with those of the other rules.
Corollary~\ref{cor:pairs} applies, for example, to the gradient of $\log \transdens_{t+1}(x_t, x_{t+1})$ with respect to a parameter of the transition densities, whose expectations under $\target$ are the terms of the score in Fisher's identity~\citep{cappe2005inference}.
It does not order the lag-one autocovariances of sums of such functions over $t$, which involve products of functions at different time steps.

\subsection{The final forced move and the gains at every step}\label{sec:gains}

The recursion~\eqref{eq:qt} explains the difference between \FMBS{} and the two rules which we propose.
For \FMBS{}, the final forced move only affects $\expprod_t(\varphi)$ through the value of $\expprod_T$ at the function $\bsprob_{T-1} \dotsm \bsprob_t \varphi$, because $\expprod_t(\varphi) = \expprod_T(\bsprob_{T-1} \dotsm \bsprob_t \varphi)$.
The backward-sampling kernels contract this function towards a constant.
We measure this contraction by the Dobrushin coefficients~\citep{dobrushin1956central}
\begin{align*}
  \dobrushin_s \coloneqq \max_{j, j' \in [N]} \frac{1}{2} \sum_{i=1}^N |\bsprob_s(i \mid j) - \bsprob_s(i \mid j')|, \quad s \in [T-1].
\end{align*}

\begin{proposition}\label{prop:washout}
Let $t \in [T]$.
Then
\begin{align*}
  0 & \leq \bbP_{\FMBS{}}(\tilde{x}_t \neq x_t) - \bbP_{\BS{}}(\tilde{x}_t \neq x_t) 
  \leq {\bbE\Biggl(\prod_{s=t}^{T-1} \dobrushin_s\Biggr)},
\end{align*}
where the expectation is taken under $\exttarget$.
If there is $\mixconst \in (0, 1]$ such that $\transdens_{s+1}(x, \cdot) \geq \mixconst \transdens_{s+1}(x', \cdot)$ for all $x, x' \in \statespace$ and $s \in [T-1]$, then the right-hand side is at most $(1 - \mixconst)^{T-t}$.
\end{proposition}

Proposition~\ref{prop:washout} is related to the literature on the forgetting of particle filters~\citep{delmoral2004feynman, karjalainen2023forgetting}.
The closest result that we know of is \citet[Lemmas~14 and~15 of the supplementary material]{karjalainen2026mixing}, who bound, under a strong mixing condition, a decoupling probability by a term which decays geometrically in the number of time steps plus a term of order $1/N$.

By Proposition~\ref{prop:washout}, unless the Dobrushin coefficients $\dobrushin_s$ are close to one, a final forced move only matters at the last time steps.
When the model factorizes across time, that is, when $\transdens_{s+1}(x, x')$ does not depend on $x$, the condition of Proposition~\ref{prop:washout} holds with $\mixconst = 1$, so that a final forced move does not change the update rates at the times $t < T$.
For \MBS{} and \CMBS{}, by contrast, the recursion~\eqref{eq:qt} adds the nonpositive term $\sum_{j,j'} \bbP_{\pos}(\refidx_{t+1} = j, \outidx_{t+1} = j') \excess_{t,j,j'}(\varphi)$ at every time $t \in [T-1]$.
Proposition~\ref{prop:gain} bounds the resulting gain in update rate from below.

\begin{proposition}\label{prop:gain}
Let $t \in [T-1]$ and assume that the particles $x_t^1, \dotsc, x_t^N$ are almost surely pairwise distinct under $\exttarget$.
For $j, j' \in [N]$, let
\begin{align*}
  \gain_t(j, j') \coloneqq \sum_{i=1}^N \bsprob_t(i \mid j) \bigl\{\bsprob_t(i \mid j') - \kernel_{t,j,j'}(i, i)\bigr\},
\end{align*}
with the kernels $\kernel_{t,j,j'}$ of \MBS{}.
Then $\gain_t(j, j') \geq 0$ and
\begin{align*}
  & \bbP_{\MBS{}}(\tilde{x}_t \neq x_t) - \bbP_{\FMBS{}}(\tilde{x}_t \neq x_t) \\
  & \quad \geq \bbE_{\MBS{}}\bigl\{\gain_t(\refidx_{t+1}, \outidx_{t+1})\bigr\}.
\end{align*}
The same holds with \CMBS{} in place of \MBS{}, in both the kernels and the inequality.
\end{proposition}

The quantity $\gain_t(j, j')$ is the decrease of the probability that the output and the reference coincide at time $t$, given $\refidx_{t+1} = j$ and $\outidx_{t+1} = j'$, when the draw of $\outidx_t$ by \FMBS{} is replaced by that of \MBS{} or \CMBS{}.
For $N = 2$, $\gain_t(j, j) = 2\{1 - \max_i \bsprob_t(i \mid j)\}^2$ for both rules.
For $j \neq j'$, $\gain_t(j, j') = 0$ for \MBS{}, whose gain is therefore made at the times at which the output and the reference coincide at the next time step, that is, when the sampler struggles to move.
For \CMBS{}, the gain is also made at the other times, with $\gain_t(j, j') = 2 \min_i \bsprob_t(i \mid j) \bsprob_t(i \mid j')$ for $N = 2$ and $j \neq j'$.

\subsection{Relationship with independent Metropolis--Hastings}\label{sec:imh}

With two particles, $\coupling_{\mathbf{p},\mathbf{p}} = \forced_{\mathbf{p}}$ for every probability vector $\mathbf{p}$, so that \CMBS{} recovers the independent Metropolis--Hastings sampler at each time step when the model factorizes across time, whatever the weights.

\begin{proposition}\label{prop:imh}
Let $N = 2$, and assume that, for $s \in [T-1]$, $\transdens_{s+1}(x, x')$ does not depend on $x$.
Then, under \CMBS{} and conditionally on $\pos$ and $\refidx_{1:T}$, the indices $\outidx_1, \dotsc, \outidx_T$ are independent, with $\outidx_t \sim \forced_{\weightvec_t}(\refidx_t, \cdot)$.
\end{proposition}

With $x_{1:T} \sim \target$ and $\refidx_{1:T}$ given, the particles $x_t^n$, $n \neq \refidx_t$, are independent draws from $\transdens_t$, whose weights relative to $x_t$ are $\weight_t^n / \weight_t^{\refidx_t} = \pot_t(x_t^n) / \pot_t(x_t)$.
By Proposition~\ref{prop:imh}, one iteration of CSMC with \CMBS{} therefore replaces each state $x_t$, independently, by the proposal $x_t^n$ with probability $\min\{1, \pot_t(x_t^n) / \pot_t(x_t)\}$: this is the independent Metropolis--Hastings sampler with proposal $\transdens_t$ at each time step.
With \MBS{}, by contrast, the index at time $t$ is drawn independently of $\refidx_t$ if $\outidx_{t+1} \neq \refidx_{t+1}$, as with \BS{}.

If the probabilities $\weightvec_T$ and $\bsprob_t(\cdot \mid j)$, $t \in [T-1]$, $j \in [N]$, are uniform, as when the potential functions are constant and $\transdens_{s+1}(x, x')$ does not depend on $x$, then the kernels $\forced_{\weightvec_T}$ and $\forced_{\bsprob_t(\cdot \mid j)}$, and, for \CMBS{}, the kernels $\kernel_{t,j,j'}$, never keep the current index, by~\eqref{eq:forced} and~\eqref{eq:coupling}.
By Lemma~\ref{lem:pair}, $\outidx_t \neq \refidx_t$ almost surely for every $t \in [T]$ under \CMBS{}, whereas $\bbP_{\pos}(\outidx_t = \refidx_t) = 1/N$ under \BS{}, and under \FMBS{} for $t < T$.
Under \MBS{}, also by Lemma~\ref{lem:pair}, $\bbP_{\pos}(\outidx_T = \refidx_T) = 0$ and, for $t < T$, the probability of $\{\outidx_t = \refidx_t\}$ is that of $\{\outidx_{t+1} \neq \refidx_{t+1}\}$ divided by $N$, so that
\begin{align*}
  \bbP_{\pos}(\outidx_t = \refidx_t) = \frac{1 - (-1/N)^{T-t}}{N + 1}.
\end{align*}
If the particles are pairwise distinct, every state is therefore updated under \CMBS{}, whereas each state is kept with probability $1/N$ under \BS{} and, under \MBS{}, with a probability which tends to $1/(N + 1)$ as $T - t$ grows.
With two particles, far from $T$, the moves of \MBS{} thus follow the Metropolis rule at about a third of the time steps, and the rule of \citet{barker1965monte} at the others.

\section{Illustration}\label{sec:numerics}

We consider two state-space models.
The first is a linear-Gaussian (LG) model with $x_1 \sim \mathcal{N}(0, \sigma^2 (1 - \rho^2)^{-1} I_d)$, $x_t \mid x_{t-1} \sim \mathcal{N}(\rho x_{t-1}, \sigma^2 I_d)$, and $y_t \mid x_t \sim \mathcal{N}(x_t, \sigma^2 I_d)$, where $\rho = 0.9$, $\sigma = 1$, $T = 200$, and $d \in \{1, 3\}$.
The second is the nonlinear (NL) model of \citet{gordon1993novel}, in the form used by \citet[Section~3.1]{andrieu2010particle}: $x_1 \sim \mathcal{N}(0, 5)$,
\begin{align*}
  x_t \mid x_{t-1} & \sim \mathcal{N}\Bigl(\frac{x_{t-1}}{2} + \frac{25 x_{t-1}}{1 + x_{t-1}^2} + 8 \cos(1.2 t), \sigma_V^2\Bigr), \\
  y_t \mid x_t & \sim \mathcal{N}\Bigl(\frac{x_t^2}{20}, \sigma_W^2\Bigr),
\end{align*}
with $\sigma_V^2 = 10$, $\sigma_W^2 \in \{10, 1\}$, and $T = 100$, as in \citet{andrieu2010particle}.
Algorithm~\ref{alg:csmc} uses the bootstrap filter, with $\transdens_t$ the transition density and $\pot_t(x_t) = \obsdens_t(y_t \mid x_t)$ the observation density, so that the potential functions only depend on the current state.

We estimate the quantities compared in Section~\ref{sec:comparison}, which concern one iteration started from the target distribution.
For each model and $N \in \{2, 4, 8, 16, 32, 64\}$, we draw $10^6$ data sets $(x_{1:T}, y_{1:T})$ from the model, so that, given $y_{1:T}$, the trajectory $x_{1:T}$ is distributed according to the smoothing distribution.
For each data set, we run steps~\ref{alg:csmc:indices} and~\ref{alg:csmc:forward} of Algorithm~\ref{alg:csmc} once with the reference $x_{1:T}$, apply the selection rules \BS{}, \FMBS{}, \MBS{}, and \CMBS{} to the same particles, and record, for each rule and each $t$, the indicator of $\{\outidx_t \neq \refidx_t\}$ and the squared jump $\lVert \tilde{x}_t - x_t \rVert^2$.
Because the particles are almost surely pairwise distinct, $\{\outidx_t \neq \refidx_t\} = \{\tilde{x}_t \neq x_t\}$ almost surely, so that the averages of these quantities over the data sets estimate the update rates and the ESJD, averaged over the distribution of the data under the model.
Because the orderings of Section~\ref{sec:comparison} hold for every data set, they also hold for the expectations which these averages estimate.
Standard errors are computed from $2{,}000$ batches of $500$ data sets, with the delta method for the ratios.
Code to reproduce the experiments can be found at \url{https://github.com/AdrienCorenflos/metropolized-backward-sampling-experiments}.
We do not show \AT{}, which \BS{} dominates (Theorem~\ref{thm:ordering} and \citealp[Theorem~11]{chopin2015particle}).

\begin{figure*}[t!]
\centering
\includegraphics[width=\textwidth]{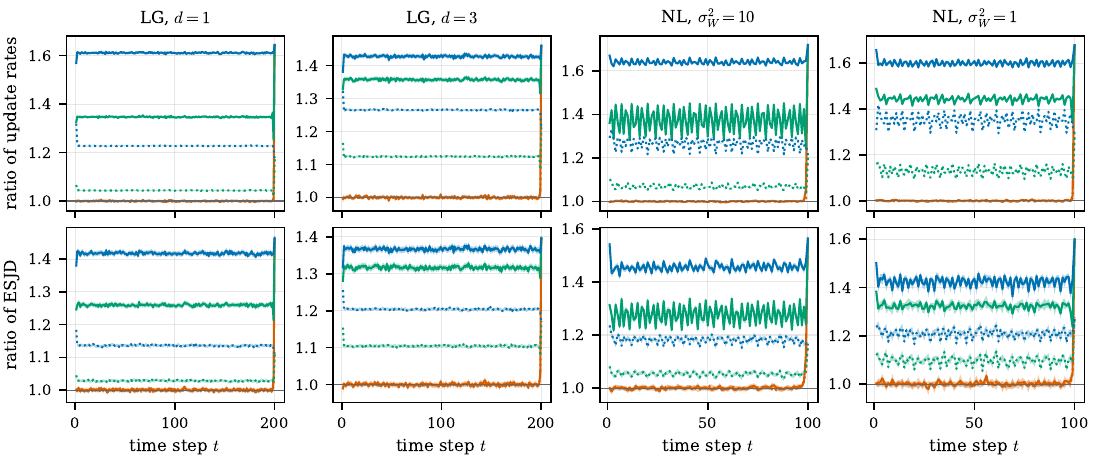}
\caption{Update rates (top) and ESJD (bottom) of forced move plus backward sampling (\FMBS{}), Metropolized backward sampling (\MBS{}), and coupled Metropolized backward sampling (\CMBS{}), relative to those of backward sampling (\BS{}), at each time step $t$, on the linear-Gaussian (LG) model with $T = 200$ and the nonlinear (NL) model with $T = 100$: \FMBS{} with $N = 2$ (\legFMBSNtwo) and $N = 8$ (\legFMBSNeight), \MBS{} with $N = 2$ (\legMBSNtwo) and $N = 8$ (\legMBSNeight), and \CMBS{} with $N = 2$ (\legCMBSNtwo) and $N = 8$ (\legCMBSNeight).
The ratios are those of the averages over $10^6$ data sets drawn from the model, with one iteration for each data set, started from the trajectory drawn with the data.
The shaded bands extend two standard errors on each side of the ratios.
For the NL model, the ratios oscillate with the time-dependent term $8 \cos(1.2 t)$ of the transition.}
\label{fig:ratios}
\end{figure*}

In agreement with Proposition~\ref{prop:washout}, the final forced move of \FMBS{} only matters at the last time steps (Figure~\ref{fig:ratios}): it increases the update rates by $2\%$ to $72\%$ at time $T$, but by less than $1\%$ on average over the time steps.
Appendix~\ref{sec:averages} gives these averages for all the models and numbers of particles.
In agreement with Theorem~\ref{thm:ordering}, \MBS{} and \CMBS{} increase the update rates and the ESJD at every time step in Figure~\ref{fig:ratios}.
On average over the time steps, \MBS{} increases them by $35\%$ to $45\%$ and by $26\%$ to $33\%$, respectively, for $N = 2$, and by $4\%$ to $13\%$ and by $3\%$ to $10\%$ for $N = 8$.
\CMBS{} increases them by $43\%$ to $64\%$ and by $37\%$ to $46\%$ for $N = 2$, and by $23\%$ to $35\%$ and by $14\%$ to $21\%$ for $N = 8$, so that its averages exceed those of \MBS{} for all the models and numbers of particles.
With stratified resampling and ancestor sampling of the reference, the increases of the update rates for $N = 2$ are the same, $35\%$ to $45\%$ for \MBS{} and $43\%$ to $64\%$ for \CMBS{} (Appendix~\ref{sec:stratified}).
The gains decrease as $N$ grows, because \BS{} then already moves with high probability.

\section{Discussion}\label{sec:discussion}

\MBS{} replaces each draw $\outidx_t \sim \bsprob_t(\cdot \mid \outidx_{t+1})$ of \BS{} by a forced move from the ancestor $\anc_t^{\outidx_{t+1}}$, at no additional cost, and keeps the reversibility of \BS{}.
\CMBS{} also avoids the index $\refidx_t$ of the reference when $\outidx_{t+1} \neq \refidx_{t+1}$, by replacing the draw of \BS{} by a forced move from $\refidx_t$ when it equals $\refidx_t$.
Its backward pass costs at most twice as much as that of \BS{}.
Their gains are largest when the index of the reference is often selected again by \BS{}, that is, for small numbers of particles or for uneven weights, which is the regime of interest when each particle is expensive to propagate or to weight.
Small numbers of particles also arise in high dimensions, where the number of particles of CSMC can be traded for locality by proposing the particles around the reference~\citep{finke2023conditional, corenflos2024particle}: this forms the original motivation for this work.

The analysis we presented here relies on multinomial resampling and on potential functions which depend on the current state only.
Together, they make the ancestor indices conditionally independent given the particle positions (Lemma~\ref{lem:extended}), from which the backward Markov chain of Lemma~\ref{lem:pair}, and hence both the reversibility and the recursion~\eqref{eq:qt}, follow.
With potential functions which depend on two states, as for guided particle filters, or with conditional versions of other resampling schemes (\citealp[Section~5]{chopin2015particle}; \citealp{karppinen2024bridge, finke2026resampling}), this conditional independence fails.
Appendix~\ref{sec:previous-state} follows the general conditional resampling approach of \citet{finke2026resampling}: the ancestor index of each state of the reference is drawn by ancestor sampling~\citep{lindsten2014particle} in the forward pass, without additional cost.
Conditionally on the particle positions and the ancestor indices, the backward Markov chain of Lemma~\ref{lem:pair} is then recovered, so that the reversibility and the results of Sections~\ref{sec:update-rates} and~\ref{sec:gains} extend to these settings.
Without ancestor sampling, the forced moves of \FMBS{} and \MBS{} from the ancestors still leave the target distribution invariant, which extends the duality of backward and ancestor sampling~\citep[Proposition~3.27]{finke2015extended} to forced moves.
On the linear-Gaussian model of Section~\ref{sec:numerics} with the locally optimal proposal~\citep{doucet2000sequential}, whose potential functions depend on the previous state, and with ancestor sampling, the update rates and the ESJD of the four selection rules are ordered as in Theorem~\ref{thm:ordering}, and \CMBS{} updates the states $49\%$ to $66\%$ more often than \BS{} with two particles (Appendix~\ref{sec:locally-optimal}).
Theorem~\ref{thm:ordering} and Corollary~\ref{cor:pairs} concern functions of the states at one or two consecutive time steps, except for the autocovariances of \CMBS{} at two time steps: whether \MBS{} or \CMBS{} dominates \BS{} for all square-integrable functions of the trajectory, which would order the asymptotic variances~\citep{tierney1998note, mira2001ordering}, remains open, as does the comparison of \MBS{} and \CMBS{}.

As a rule of thumb, we recommend \MBS{}.
It dominates \BS{} in the sense of Theorem~\ref{thm:ordering} at exactly the same cost.
For potential functions of the previous state and other resampling schemes, it remains valid with the same forward pass, whereas \CMBS{} then requires ancestor sampling of the reference and therefore a different implementation of the forward pass.
The backward pass of \CMBS{} also costs up to twice as much as that of \BS{}.
\CMBS{} is worth its additional cost when the potential functions are much more expensive to evaluate than the transition densities, because the forward pass then dominates the cost of an iteration, while the backward pass only evaluates transition densities.
This is the case, for example, in approximate Bayesian computation for hidden Markov models whose observation densities are intractable but can be simulated~\citep{jasra2012filtering, dean2014parameter}.

Finally, \MBS{} and \CMBS{} only change how each backward index is drawn.
They can therefore be combined with other improvements of backward sampling, such as the bridge backward sampling of \citet{karppinen2024bridge}.
By the duality of backward and ancestor sampling, the forced moves of \MBS{} correspond to forced moves of the ancestor index $\anc_t^{\refidx_{t+1}}$ of the reference, started from $\refidx_t$, which are allowed in particle Gibbs with ancestor sampling~\citep[Section~6.1]{lindsten2014particle}.

\bibliographystyle{imsart-nameyear}
\bibliography{refs}

\begin{thebibliography}{33}

\bibitem[\protect\citeauthoryear{Andrieu, Doucet and
  Holenstein}{2010}]{andrieu2010particle}
\begin{barticle}[author]
\bauthor{\bsnm{Andrieu},~\bfnm{Christophe}\binits{C.}},
  \bauthor{\bsnm{Doucet},~\bfnm{Arnaud}\binits{A.}} \AND
  \bauthor{\bsnm{Holenstein},~\bfnm{Roman}\binits{R.}}
(\byear{2010}).
\btitle{Particle {M}arkov chain {M}onte {C}arlo methods}.
\bjournal{Journal of the Royal Statistical Society: Series B (Statistical
  Methodology)}
\bvolume{72}
\bpages{269--342}.
\bnote{With discussion}.
\bdoi{10.1111/j.1467-9868.2009.00736.x}
\end{barticle}
\endbibitem

\bibitem[\protect\citeauthoryear{Andrieu, Lee and
  Vihola}{2018}]{andrieu2018uniform}
\begin{barticle}[author]
\bauthor{\bsnm{Andrieu},~\bfnm{Christophe}\binits{C.}},
  \bauthor{\bsnm{Lee},~\bfnm{Anthony}\binits{A.}} \AND
  \bauthor{\bsnm{Vihola},~\bfnm{Matti}\binits{M.}}
(\byear{2018}).
\btitle{Uniform ergodicity of the iterated conditional {SMC} and geometric
  ergodicity of particle {G}ibbs samplers}.
\bjournal{Bernoulli}
\bvolume{24}
\bpages{842--872}.
\bdoi{10.3150/15-BEJ785}
\end{barticle}
\endbibitem

\bibitem[\protect\citeauthoryear{Barker}{1965}]{barker1965monte}
\begin{barticle}[author]
\bauthor{\bsnm{Barker},~\bfnm{A.~A.}\binits{A.~A.}}
(\byear{1965}).
\btitle{{M}onte {C}arlo calculations of the radial distribution functions for a
  proton--electron plasma}.
\bjournal{Australian Journal of Physics}
\bvolume{18}
\bpages{119--134}.
\bdoi{10.1071/PH650119}
\end{barticle}
\endbibitem

\bibitem[\protect\citeauthoryear{Bunch and Godsill}{2013}]{bunch2013improved}
\begin{barticle}[author]
\bauthor{\bsnm{Bunch},~\bfnm{Pete}\binits{P.}} \AND
  \bauthor{\bsnm{Godsill},~\bfnm{Simon}\binits{S.}}
(\byear{2013}).
\btitle{Improved particle approximations to the joint smoothing distribution
  using {M}arkov chain {M}onte {C}arlo}.
\bjournal{IEEE Transactions on Signal Processing}
\bvolume{61}
\bpages{956--963}.
\bdoi{10.1109/TSP.2012.2229277}
\end{barticle}
\endbibitem

\bibitem[\protect\citeauthoryear{Capp{\'e}, Moulines and
  Ryd{\'e}n}{2005}]{cappe2005inference}
\begin{bbook}[author]
\bauthor{\bsnm{Capp{\'e}},~\bfnm{Olivier}\binits{O.}},
  \bauthor{\bsnm{Moulines},~\bfnm{Eric}\binits{E.}} \AND
  \bauthor{\bsnm{Ryd{\'e}n},~\bfnm{Tobias}\binits{T.}}
(\byear{2005}).
\btitle{Inference in Hidden {M}arkov Models}.
\bseries{Springer Series in Statistics}.
\bpublisher{Springer}, \baddress{New York}.
\bdoi{10.1007/0-387-28982-8}
\end{bbook}
\endbibitem

\bibitem[\protect\citeauthoryear{Chopin and
  Papaspiliopoulos}{2020}]{chopin2020introduction}
\begin{bbook}[author]
\bauthor{\bsnm{Chopin},~\bfnm{Nicolas}\binits{N.}} \AND
  \bauthor{\bsnm{Papaspiliopoulos},~\bfnm{Omiros}\binits{O.}}
(\byear{2020}).
\btitle{An Introduction to Sequential {M}onte {C}arlo}.
\bseries{Springer Series in Statistics}.
\bpublisher{Springer}, \baddress{Cham}.
\bdoi{10.1007/978-3-030-47845-2}
\end{bbook}
\endbibitem

\bibitem[\protect\citeauthoryear{Chopin and Singh}{2013}]{chopin2013particle}
\begin{bmisc}[author]
\bauthor{\bsnm{Chopin},~\bfnm{Nicolas}\binits{N.}} \AND
  \bauthor{\bsnm{Singh},~\bfnm{Sumeetpal~S.}\binits{S.~S.}}
(\byear{2013}).
\btitle{On the particle {G}ibbs sampler}.
\bhowpublished{arXiv:1304.1887v1}.
\bnote{First version of \citet{chopin2015particle}}.
\end{bmisc}
\endbibitem

\bibitem[\protect\citeauthoryear{Chopin and Singh}{2015}]{chopin2015particle}
\begin{barticle}[author]
\bauthor{\bsnm{Chopin},~\bfnm{Nicolas}\binits{N.}} \AND
  \bauthor{\bsnm{Singh},~\bfnm{Sumeetpal~S.}\binits{S.~S.}}
(\byear{2015}).
\btitle{On particle {G}ibbs sampling}.
\bjournal{Bernoulli}
\bvolume{21}
\bpages{1855--1883}.
\bdoi{10.3150/14-BEJ629}
\end{barticle}
\endbibitem

\bibitem[\protect\citeauthoryear{Corenflos and
  Finke}{2024}]{corenflos2024particle}
\begin{bmisc}[author]
\bauthor{\bsnm{Corenflos},~\bfnm{Adrien}\binits{A.}} \AND
  \bauthor{\bsnm{Finke},~\bfnm{Axel}\binits{A.}}
(\byear{2024}).
\btitle{Particle-{MALA} and {P}article-m{GRAD}: Gradient-based {MCMC} methods
  for high-dimensional state-space models}.
\bhowpublished{arXiv:2401.14868v1}.
\end{bmisc}
\endbibitem

\bibitem[\protect\citeauthoryear{Dau and Chopin}{2023}]{dau2023backward}
\begin{barticle}[author]
\bauthor{\bsnm{Dau},~\bfnm{Hai-Dang}\binits{H.-D.}} \AND
  \bauthor{\bsnm{Chopin},~\bfnm{Nicolas}\binits{N.}}
(\byear{2023}).
\btitle{On backward smoothing algorithms}.
\bjournal{The Annals of Statistics}
\bvolume{51}
\bpages{2145--2169}.
\bdoi{10.1214/23-AOS2324}
\end{barticle}
\endbibitem

\bibitem[\protect\citeauthoryear{Dean et~al.}{2014}]{dean2014parameter}
\begin{barticle}[author]
\bauthor{\bsnm{Dean},~\bfnm{Thomas~A.}\binits{T.~A.}},
  \bauthor{\bsnm{Singh},~\bfnm{Sumeetpal~S.}\binits{S.~S.}},
  \bauthor{\bsnm{Jasra},~\bfnm{Ajay}\binits{A.}} \AND
  \bauthor{\bsnm{Peters},~\bfnm{Gareth~W.}\binits{G.~W.}}
(\byear{2014}).
\btitle{Parameter estimation for hidden {M}arkov models with intractable
  likelihoods}.
\bjournal{Scandinavian Journal of Statistics}
\bvolume{41}
\bpages{970--987}.
\bdoi{10.1111/sjos.12077}
\end{barticle}
\endbibitem

\bibitem[\protect\citeauthoryear{Del~Moral}{2004}]{delmoral2004feynman}
\begin{bbook}[author]
\bauthor{\bsnm{Del~Moral},~\bfnm{Pierre}\binits{P.}}
(\byear{2004}).
\btitle{{F}eynman--{K}ac Formulae: Genealogical and Interacting Particle
  Systems with Applications}.
\bseries{Probability and its Applications}.
\bpublisher{Springer}, \baddress{New York}.
\bdoi{10.1007/978-1-4684-9393-1}
\end{bbook}
\endbibitem

\bibitem[\protect\citeauthoryear{Dobrushin}{1956}]{dobrushin1956central}
\begin{barticle}[author]
\bauthor{\bsnm{Dobrushin},~\bfnm{R.~L.}\binits{R.~L.}}
(\byear{1956}).
\btitle{Central limit theorem for nonstationary {M}arkov chains. {I}}.
\bjournal{Theory of Probability and Its Applications}
\bvolume{1}
\bpages{65--80}.
\bdoi{10.1137/1101006}
\end{barticle}
\endbibitem

\bibitem[\protect\citeauthoryear{Doucet, Godsill and
  Andrieu}{2000}]{doucet2000sequential}
\begin{barticle}[author]
\bauthor{\bsnm{Doucet},~\bfnm{Arnaud}\binits{A.}},
  \bauthor{\bsnm{Godsill},~\bfnm{Simon}\binits{S.}} \AND
  \bauthor{\bsnm{Andrieu},~\bfnm{Christophe}\binits{C.}}
(\byear{2000}).
\btitle{On sequential {M}onte {C}arlo sampling methods for {B}ayesian
  filtering}.
\bjournal{Statistics and Computing}
\bvolume{10}
\bpages{197--208}.
\bdoi{10.1023/A:1008935410038}
\end{barticle}
\endbibitem

\bibitem[\protect\citeauthoryear{Finke}{2015}]{finke2015extended}
\begin{bphdthesis}[author]
\bauthor{\bsnm{Finke},~\bfnm{Axel}\binits{A.}}
(\byear{2015}).
\btitle{On extended state-space constructions for {M}onte {C}arlo methods},
\btype{PhD thesis},
\bpublisher{University of Warwick}.
\end{bphdthesis}
\endbibitem

\bibitem[\protect\citeauthoryear{Finke and Thiery}{2023}]{finke2023conditional}
\begin{barticle}[author]
\bauthor{\bsnm{Finke},~\bfnm{Axel}\binits{A.}} \AND
  \bauthor{\bsnm{Thiery},~\bfnm{Alexandre~H.}\binits{A.~H.}}
(\byear{2023}).
\btitle{Conditional sequential {M}onte {C}arlo in high dimensions}.
\bjournal{The Annals of Statistics}
\bvolume{51}
\bpages{437--463}.
\bdoi{10.1214/22-AOS2252}
\end{barticle}
\endbibitem

\bibitem[\protect\citeauthoryear{Finke et~al.}{2026}]{finke2026resampling}
\begin{bmisc}[author]
\bauthor{\bsnm{Finke},~\bfnm{Axel}\binits{A.}},
  \bauthor{\bsnm{Johansen},~\bfnm{Adam~M.}\binits{A.~M.}},
  \bauthor{\bsnm{Lee},~\bfnm{Anthony}\binits{A.}} \AND
  \bauthor{\bsnm{Murray},~\bfnm{Lawrence~M.}\binits{L.~M.}}
(\byear{2026}).
\btitle{Resampling in conditional {SMC} algorithms}.
\bhowpublished{arXiv:2606.25603v1}.
\end{bmisc}
\endbibitem

\bibitem[\protect\citeauthoryear{Godsill, Doucet and
  West}{2004}]{godsill2004monte}
\begin{barticle}[author]
\bauthor{\bsnm{Godsill},~\bfnm{Simon~J.}\binits{S.~J.}},
  \bauthor{\bsnm{Doucet},~\bfnm{Arnaud}\binits{A.}} \AND
  \bauthor{\bsnm{West},~\bfnm{Mike}\binits{M.}}
(\byear{2004}).
\btitle{{M}onte {C}arlo smoothing for nonlinear time series}.
\bjournal{Journal of the American Statistical Association}
\bvolume{99}
\bpages{156--168}.
\bdoi{10.1198/016214504000000151}
\end{barticle}
\endbibitem

\bibitem[\protect\citeauthoryear{Goldstein}{1979}]{goldstein1979maximal}
\begin{barticle}[author]
\bauthor{\bsnm{Goldstein},~\bfnm{Sheldon}\binits{S.}}
(\byear{1979}).
\btitle{Maximal coupling}.
\bjournal{Zeitschrift f{\"u}r Wahrscheinlichkeitstheorie und Verwandte Gebiete}
\bvolume{46}
\bpages{193--204}.
\bdoi{10.1007/BF00533259}
\end{barticle}
\endbibitem

\bibitem[\protect\citeauthoryear{Gordon, Salmond and
  Smith}{1993}]{gordon1993novel}
\begin{barticle}[author]
\bauthor{\bsnm{Gordon},~\bfnm{Neil~J.}\binits{N.~J.}},
  \bauthor{\bsnm{Salmond},~\bfnm{David~J.}\binits{D.~J.}} \AND
  \bauthor{\bsnm{Smith},~\bfnm{Adrian F.~M.}\binits{A.~F.~M.}}
(\byear{1993}).
\btitle{Novel approach to nonlinear/non-{G}aussian {B}ayesian state
  estimation}.
\bjournal{IEE Proceedings F (Radar and Signal Processing)}
\bvolume{140}
\bpages{107--113}.
\bdoi{10.1049/ip-f-2.1993.0015}
\end{barticle}
\endbibitem

\bibitem[\protect\citeauthoryear{Jasra et~al.}{2012}]{jasra2012filtering}
\begin{barticle}[author]
\bauthor{\bsnm{Jasra},~\bfnm{Ajay}\binits{A.}},
  \bauthor{\bsnm{Singh},~\bfnm{Sumeetpal~S.}\binits{S.~S.}},
  \bauthor{\bsnm{Martin},~\bfnm{James~S.}\binits{J.~S.}} \AND
  \bauthor{\bsnm{McCoy},~\bfnm{Emma}\binits{E.}}
(\byear{2012}).
\btitle{Filtering via approximate {B}ayesian computation}.
\bjournal{Statistics and Computing}
\bvolume{22}
\bpages{1223--1237}.
\bdoi{10.1007/s11222-010-9185-0}
\end{barticle}
\endbibitem

\bibitem[\protect\citeauthoryear{Karjalainen
  et~al.}{2023}]{karjalainen2023forgetting}
\begin{bmisc}[author]
\bauthor{\bsnm{Karjalainen},~\bfnm{Joona}\binits{J.}},
  \bauthor{\bsnm{Lee},~\bfnm{Anthony}\binits{A.}},
  \bauthor{\bsnm{Singh},~\bfnm{Sumeetpal~S.}\binits{S.~S.}} \AND
  \bauthor{\bsnm{Vihola},~\bfnm{Matti}\binits{M.}}
(\byear{2023}).
\btitle{On the forgetting of particle filters}.
\bhowpublished{arXiv:2309.08517v4}.
\end{bmisc}
\endbibitem

\bibitem[\protect\citeauthoryear{Karjalainen
  et~al.}{2026}]{karjalainen2026mixing}
\begin{barticle}[author]
\bauthor{\bsnm{Karjalainen},~\bfnm{Joona}\binits{J.}},
  \bauthor{\bsnm{Lee},~\bfnm{Anthony}\binits{A.}},
  \bauthor{\bsnm{Singh},~\bfnm{Sumeetpal~S.}\binits{S.~S.}} \AND
  \bauthor{\bsnm{Vihola},~\bfnm{Matti}\binits{M.}}
(\byear{2026}).
\btitle{Mixing time of the conditional backward sampling particle filter}.
\bjournal{Journal of the Royal Statistical Society Series B: Statistical
  Methodology}
\bvolume{88}
\bpages{1087--1106}.
\bdoi{10.1093/jrsssb/qkaf078}
\end{barticle}
\endbibitem

\bibitem[\protect\citeauthoryear{Karppinen, Singh and
  Vihola}{2024}]{karppinen2024bridge}
\begin{barticle}[author]
\bauthor{\bsnm{Karppinen},~\bfnm{Santeri}\binits{S.}},
  \bauthor{\bsnm{Singh},~\bfnm{Sumeetpal~S.}\binits{S.~S.}} \AND
  \bauthor{\bsnm{Vihola},~\bfnm{Matti}\binits{M.}}
(\byear{2024}).
\btitle{Conditional particle filters with bridge backward sampling}.
\bjournal{Journal of Computational and Graphical Statistics}
\bvolume{33}
\bpages{364--378}.
\bdoi{10.1080/10618600.2023.2231514}
\end{barticle}
\endbibitem

\bibitem[\protect\citeauthoryear{Kitagawa}{1996}]{kitagawa1996monte}
\begin{barticle}[author]
\bauthor{\bsnm{Kitagawa},~\bfnm{Genshiro}\binits{G.}}
(\byear{1996}).
\btitle{Monte {C}arlo Filter and Smoother for Non-{G}aussian Nonlinear State
  Space Models}.
\bjournal{Journal of Computational and Graphical Statistics}
\bvolume{5}
\bpages{1--25}.
\bdoi{10.1080/10618600.1996.10474692}
\end{barticle}
\endbibitem

\bibitem[\protect\citeauthoryear{Lee, Singh and Vihola}{2020}]{lee2020coupled}
\begin{barticle}[author]
\bauthor{\bsnm{Lee},~\bfnm{Anthony}\binits{A.}},
  \bauthor{\bsnm{Singh},~\bfnm{Sumeetpal~S.}\binits{S.~S.}} \AND
  \bauthor{\bsnm{Vihola},~\bfnm{Matti}\binits{M.}}
(\byear{2020}).
\btitle{Coupled conditional backward sampling particle filter}.
\bjournal{The Annals of Statistics}
\bvolume{48}
\bpages{3066--3089}.
\bdoi{10.1214/19-AOS1922}
\end{barticle}
\endbibitem

\bibitem[\protect\citeauthoryear{Lindsten, Jordan and
  Sch{\"o}n}{2014}]{lindsten2014particle}
\begin{barticle}[author]
\bauthor{\bsnm{Lindsten},~\bfnm{Fredrik}\binits{F.}},
  \bauthor{\bsnm{Jordan},~\bfnm{Michael~I.}\binits{M.~I.}} \AND
  \bauthor{\bsnm{Sch{\"o}n},~\bfnm{Thomas~B.}\binits{T.~B.}}
(\byear{2014}).
\btitle{Particle {G}ibbs with ancestor sampling}.
\bjournal{Journal of Machine Learning Research}
\bvolume{15}
\bpages{2145--2184}.
\end{barticle}
\endbibitem

\bibitem[\protect\citeauthoryear{Lindsten and
  Sch{\"o}n}{2013}]{lindsten2013backward}
\begin{barticle}[author]
\bauthor{\bsnm{Lindsten},~\bfnm{Fredrik}\binits{F.}} \AND
  \bauthor{\bsnm{Sch{\"o}n},~\bfnm{Thomas~B.}\binits{T.~B.}}
(\byear{2013}).
\btitle{Backward simulation methods for {M}onte {C}arlo statistical inference}.
\bjournal{Foundations and Trends in Machine Learning}
\bvolume{6}
\bpages{1--143}.
\bdoi{10.1561/2200000045}
\end{barticle}
\endbibitem

\bibitem[\protect\citeauthoryear{Liu}{1996}]{liu1996peskun}
\begin{barticle}[author]
\bauthor{\bsnm{Liu},~\bfnm{Jun~S.}\binits{J.~S.}}
(\byear{1996}).
\btitle{Peskun's theorem and a modified discrete-state {G}ibbs sampler}.
\bjournal{Biometrika}
\bvolume{83}
\bpages{681--682}.
\bdoi{10.1093/biomet/83.3.681}
\end{barticle}
\endbibitem

\bibitem[\protect\citeauthoryear{Mira}{2001}]{mira2001ordering}
\begin{barticle}[author]
\bauthor{\bsnm{Mira},~\bfnm{Antonietta}\binits{A.}}
(\byear{2001}).
\btitle{Ordering and improving the performance of {M}onte {C}arlo {M}arkov
  chains}.
\bjournal{Statistical Science}
\bvolume{16}
\bpages{340--350}.
\bdoi{10.1214/ss/1015346319}
\end{barticle}
\endbibitem

\bibitem[\protect\citeauthoryear{Peskun}{1973}]{peskun1973optimum}
\begin{barticle}[author]
\bauthor{\bsnm{Peskun},~\bfnm{Peter~H.}\binits{P.~H.}}
(\byear{1973}).
\btitle{Optimum {M}onte-{C}arlo sampling using {M}arkov chains}.
\bjournal{Biometrika}
\bvolume{60}
\bpages{607--612}.
\bdoi{10.1093/biomet/60.3.607}
\end{barticle}
\endbibitem

\bibitem[\protect\citeauthoryear{Tierney}{1998}]{tierney1998note}
\begin{barticle}[author]
\bauthor{\bsnm{Tierney},~\bfnm{Luke}\binits{L.}}
(\byear{1998}).
\btitle{A note on {M}etropolis--{H}astings kernels for general state spaces}.
\bjournal{The Annals of Applied Probability}
\bvolume{8}
\bpages{1--9}.
\bdoi{10.1214/aoap/1027961031}
\end{barticle}
\endbibitem

\bibitem[\protect\citeauthoryear{Whiteley}{2010}]{whiteley2010particle}
\begin{barticle}[author]
\bauthor{\bsnm{Whiteley},~\bfnm{Nick}\binits{N.}}
(\byear{2010}).
\btitle{Contribution to the discussion on `{P}article {M}arkov chain {M}onte
  {C}arlo methods' by {A}ndrieu, {C}., {D}oucet, {A}., and {H}olenstein, {R}.}
\bjournal{Journal of the Royal Statistical Society: Series B (Statistical
  Methodology)}
\bvolume{72}
\bpages{306--307}.
\bdoi{10.1111/j.1467-9868.2009.00736.x}
\end{barticle}
\endbibitem

\end{thebibliography}

\begin{appendix}
\section{Proofs}\label{sec:proofs}

This appendix gives the proofs of our results, in the order in which they appear.

\begin{proof}[Proof of Lemma~\ref{lem:coupling}]
Let $\mathbf{u} \coloneqq \mathbf{p} \odot \mathbf{q}$, so that $\mathbf{q} \odot \mathbf{p} = \mathbf{u}$.
By the first part of Lemma~\ref{lem:forced} if $\mathbf{u} \neq \mathbf{0}$, and trivially if $\mathbf{u} = \mathbf{0}$, $u_i \forced_{\mathbf{u}}(i, i') = u_{i'} \forced_{\mathbf{u}}(i', i)$ for all $i, i' \in [N]$.
It follows from~\eqref{eq:coupling} that, for $i \neq i'$,
\begin{align*}
  p_i \coupling_{\mathbf{p},\mathbf{q}}(i, i') & = p_i q_{i'} + u_i \forced_{\mathbf{u}}(i, i') \\
  & = p_i q_{i'} + u_{i'} \forced_{\mathbf{u}}(i', i) = q_{i'} \coupling_{\mathbf{q},\mathbf{p}}(i', i),
\end{align*}
whereas $p_i \coupling_{\mathbf{p},\mathbf{q}}(i, i) = u_i \forced_{\mathbf{u}}(i, i) = q_i \coupling_{\mathbf{q},\mathbf{p}}(i, i)$.
The second part follows from $\coupling_{\mathbf{p},\mathbf{q}}(i, i') = q_{i'} + q_i \forced_{\mathbf{u}}(i, i')$ for $i \neq i'$.
For the third part, the coefficients $p_i \{\coupling_{\mathbf{p},\mathbf{q}}(i, i') - q_{i'}\}$, $i, i' \in [N]$, sum to zero over $i'$ for every $i$, because $\coupling_{\mathbf{p},\mathbf{q}}$ is a Markov kernel, and over $i$ for every $i'$, because $\sum_i p_i \coupling_{\mathbf{p},\mathbf{q}}(i, i') = q_{i'} \sum_i \coupling_{\mathbf{q},\mathbf{p}}(i', i) = q_{i'}$ by the first part.
Because $\langle \varphi, \coupling_{\mathbf{p},\mathbf{q}} \varphi \rangle_{\mathbf{p}} - \mathbf{p}(\varphi) \mathbf{q}(\varphi)$ is the sum of these coefficients times $\varphi(i) \varphi(i')$, expanding $\{\varphi(i) - \varphi(i')\}^2 = \varphi(i)^2 + \varphi(i')^2 - 2 \varphi(i) \varphi(i')$, whose first two terms contribute zero to the sum, gives
\begin{align*}
  & \langle \varphi, \coupling_{\mathbf{p},\mathbf{q}} \varphi \rangle_{\mathbf{p}} - \mathbf{p}(\varphi) \mathbf{q}(\varphi) \\
  & \quad = -\frac{1}{2} \sum_{i \neq i'} p_i \{\coupling_{\mathbf{p},\mathbf{q}}(i, i') - q_{i'}\} \{\varphi(i) - \varphi(i')\}^2,
\end{align*}
where the terms with $i = i'$ vanish.
The right-hand side is nonpositive by the second part.
\end{proof}

\begin{proof}[Proof of Lemma~\ref{lem:extended}]
This proof is not ours and is provided for completeness and ease of reference.
We collect arguments from \citet{andrieu2010particle} in the form of \citet[Section~2.3]{chopin2015particle} and \citet[proof of Proposition~9]{chopin2015particle}.
Steps~\ref{alg:csmc:indices} and~\ref{alg:csmc:forward} of Algorithm~\ref{alg:csmc} draw $\refidx_{1:T}$ with probability $N^{-T}$, then $x_1^n \sim \transdens_1$ for $n \neq \refidx_1$, and, for $2 \leq t \leq T$ and $n \neq \refidx_t$, $\anc_{t-1}^n \sim \weightvec_{t-1}$ and $x_t^n \sim \transdens_t(x_{t-1}^{\anc_{t-1}^n}, \cdot)$, with $\anc_{t-1}^{\refidx_t} = \refidx_{t-1}$.
With $x_{1:T} \sim \target$, on the set where $\anc_{t-1}^{\refidx_t} = \refidx_{t-1}$ for $2 \leq t \leq T$, the variables $(\pos, \anc, \refidx_{1:T})$ therefore have the density
\begin{align*}
  & \target(x_1^{\refidx_1}, \dotsc, x_T^{\refidx_T}) \frac{1}{N^T} \prod_{n \neq \refidx_1} \transdens_1(x_1^n) \\
  & \times \prod_{t=2}^T \prod_{n \neq \refidx_t} \weight_{t-1}^{\anc_{t-1}^n} \transdens_t(x_{t-1}^{\anc_{t-1}^n}, x_t^n).
\end{align*}
The factors $\transdens_1(x_1^{\refidx_1})$ and $\transdens_t(x_{t-1}^{\refidx_{t-1}}, x_t^{\refidx_t})$ of~\eqref{eq:target} complete the products over $n$, because $x_{t-1}^{\refidx_{t-1}} = x_{t-1}^{\anc_{t-1}^{\refidx_t}}$.
By the definition of the weights, $\pot_t(x_t^{\refidx_t}) = \weight_t^{\refidx_t} \sum_{i=1}^N \pot_t(x_t^i)$ for $t \in [T]$, and, because $\refidx_{t-1} = \anc_{t-1}^{\refidx_t}$,
\begin{align*}
  \prod_{t=1}^T \weight_t^{\refidx_t} = \weight_T^{\refidx_T} \prod_{t=2}^T \weight_{t-1}^{\anc_{t-1}^{\refidx_t}},
\end{align*}
which completes the products of the weights over $n$.
The density is therefore~\eqref{eq:extended}.
Because $\refidx_{t-1} = \anc_{t-1}^{\refidx_t}$, the indices $\refidx_{1:T-1}$ are functions of $\anc$ and $\refidx_T$, so that $(\pos, \anc, \refidx_T) \sim \exttarget$ and $x_{1:T}$ is the trajectory traced back from $\refidx_T$.
Because the weights depend on $\pos$ only, \eqref{eq:extended} is, as a function of $(\anc, \refidx_T)$, proportional to
\begin{align*}
  \weight_T^{\refidx_T} \prod_{t=2}^T \prod_{n=1}^N \weight_{t-1}^{\anc_{t-1}^n} \transdens_t(x_{t-1}^{\anc_{t-1}^n}, x_t^n),
\end{align*}
which gives the second part, by~\eqref{eq:bs}.
\end{proof}

\begin{proof}[Proof of Lemma~\ref{lem:pair}]
By Lemma~\ref{lem:extended}, $(\pos, \anc, \refidx_T) \sim \exttarget$, so that, conditionally on $\pos$, $\refidx_T \sim \weightvec_T$, and the first claim follows from \algstep{alg:csmc:final}.
Let $t \in [T-1]$.
The indices $(\refidx_{t+1:T}, \outidx_{t+1:T})$ are functions of $\refidx_T$, of the ancestor indices $\anc_s^n$ with $s > t$, and of the draws of steps~\ref{alg:csmc:final} and~\ref{alg:csmc:backward} of Algorithm~\ref{alg:csmc} at the times $s > t$, also with the draws $\outidx_s \sim \kernel_{s, \outidx_{s+1}, \outidx_{s+1}}(\anc_s^{\outidx_{s+1}}, \cdot)$.
By Lemma~\ref{lem:extended}, conditionally on $\pos$ and on these variables, the indices $\anc_t^n$, $n \in [N]$, are therefore independent, with the distributions $\bsprob_t(\cdot \mid n)$.
Hence $\refidx_t = \anc_t^{\refidx_{t+1}} \sim \bsprob_t(\cdot \mid \refidx_{t+1})$.
By \algstep{alg:csmc:backward}, given also $(\anc_t^n)_{n \in [N]}$, $\outidx_t \sim \kernel_{t, \refidx_{t+1}, \outidx_{t+1}}(\refidx_t, \cdot)$, which depends on these indices only through $\refidx_t$.
With $\outidx_t \sim \kernel_{t, \outidx_{t+1}, \outidx_{t+1}}(\anc_t^{\outidx_{t+1}}, \cdot)$ instead, $\outidx_t \sim \kernel_{t, \refidx_{t+1}, \refidx_{t+1}}(\refidx_t, \cdot)$ if $\outidx_{t+1} = \refidx_{t+1}$, because then $\anc_t^{\outidx_{t+1}} = \refidx_t$.
If $\outidx_{t+1} \neq \refidx_{t+1}$, then $\anc_t^{\outidx_{t+1}}$ is independent of $\refidx_t$, with the distribution $\bsprob_t(\cdot \mid \outidx_{t+1})$, which is invariant for $\kernel_{t, \outidx_{t+1}, \outidx_{t+1}}$, so that $\outidx_t$ is independent of $\refidx_t$, with the distribution $\bsprob_t(\cdot \mid \outidx_{t+1}) = \kernel_{t, \refidx_{t+1}, \outidx_{t+1}}(\refidx_t, \cdot)$.
In all cases, these distributions depend on the conditioning variables only through $\pos$, $\refidx_{t+1}$, and $\outidx_{t+1}$, which proves the claim.
\end{proof}

\begin{proof}[Proof of Proposition~\ref{prop:reversible}]
Let $x_{1:T} \sim \target$ and run Algorithm~\ref{alg:csmc}, so that $x_t = x_t^{\refidx_t}$ and $\tilde{x}_t = x_t^{\outidx_t}$ for $t \in [T]$.
By Lemma~\ref{lem:pair}, conditionally on $\pos$, the pairs $(\refidx_t, \outidx_t)$, $t = T, \dotsc, 1$, form a Markov chain on $[N]^2$.
Its initial distribution $\weight_T^i \kernel_T(i, j)$ is symmetric in $(i, j)$, by the reversibility of $\kernel_T$.
At time $t \in [T-1]$, its transition from $(j, j')$ to $(i, i')$ has the probability $\bsprob_t(i \mid j) \kernel_{t,j,j'}(i, i')$, which, by~\eqref{eq:symmetry}, does not change when $(i, j)$ and $(i', j')$ are exchanged.
Hence, conditionally on $\pos$, the distribution of $(\refidx_{1:T}, \outidx_{1:T})$ does not change when $\refidx_{1:T}$ and $\outidx_{1:T}$ are exchanged.
The pairs $(x_{1:T}, \tilde{x}_{1:T})$ and $(\tilde{x}_{1:T}, x_{1:T})$ hence have the same distribution, that is, the Markov kernel defined by Algorithm~\ref{alg:csmc} is reversible with respect to $\target$.
\end{proof}

\begin{proof}[Proof of Theorem~\ref{thm:ordering}]
We show, by backward induction on $t$, that, for every $\varphi \colon [N] \to \mathbb{R}$,
\begin{equation}
  \expprod_t^{\MBS{}}(\varphi) \leq \expprod_t^{\FMBS{}}(\varphi) \leq \expprod_t^{\BS{}}(\varphi) \leq \expprod_t^{\AT{}}(\varphi), \label{eq:qorder}
\end{equation}
and $\expprod_t^{\CMBS{}}(\varphi) \leq \expprod_t^{\FMBS{}}(\varphi)$, where the superscripts indicate the selection rule.
For $t = T$, by Lemma~\ref{lem:forced},
\begin{align*}
  \expprod_T^{\MBS{}}(\varphi) = \expprod_T^{\CMBS{}}(\varphi) = \expprod_T^{\FMBS{}}(\varphi) & = \langle \varphi, \forced_{\weightvec_T} \varphi \rangle_{\weightvec_T} \\
  & \leq \weightvec_T(\varphi)^2,
\end{align*}
and $\weightvec_T(\varphi)^2 = \expprod_T^{\BS{}}(\varphi) = \expprod_T^{\AT{}}(\varphi)$.
Assume that the claim holds at time $t + 1$.
By~\eqref{eq:qt} and the signs of $\excess_{t,j,j'}$, for \MBS{} and, in the same way, for \CMBS{},
\begin{align*}
  \expprod_t^{\MBS{}}(\varphi) & \leq \expprod_{t+1}^{\MBS{}}(\bsprob_t \varphi) \leq \expprod_{t+1}^{\FMBS{}}(\bsprob_t \varphi) = \expprod_t^{\FMBS{}}(\varphi), \\
  \expprod_t^{\FMBS{}}(\varphi) & = \expprod_{t+1}^{\FMBS{}}(\bsprob_t \varphi) \leq \expprod_{t+1}^{\BS{}}(\bsprob_t \varphi) = \expprod_t^{\BS{}}(\varphi), \\
  \expprod_t^{\BS{}}(\varphi) & = \expprod_{t+1}^{\BS{}}(\bsprob_t \varphi) \leq \expprod_{t+1}^{\AT{}}(\bsprob_t \varphi) \leq \expprod_t^{\AT{}}(\varphi).
\end{align*}
With $\varphi(i) \coloneqq f(x_t^i)$, we have $f(x_t) f(\tilde{x}_t) = \varphi(\refidx_t) \varphi(\outidx_t)$, so that $\bbE_{\pos}\{f(x_t) f(\tilde{x}_t)\} = \expprod_t(\varphi)$, whose expectation under the distribution of $\pos$, which does not depend on the selection rule, is the expectation in Theorem~\ref{thm:ordering}.
These expectations are finite, because $2 |f(x_t) f(\tilde{x}_t)|$ is at most $f(x_t)^2 + f(\tilde{x}_t)^2$, where $\tilde{x}_{1:T} \sim \target$ by Proposition~\ref{prop:reversible}.
Given $\pos$, the event $\{\tilde{x}_t = x_t\}$ is the disjoint union of the events $\{\refidx_t \in S, \outidx_t \in S\}$ over the classes $S \subseteq [N]$ of indices whose particles at time $t$ are equal, so that $\bbP_{\pos}(\tilde{x}_t = x_t) = \sum_S \expprod_t(\mathbf{1}_S)$, which is ordered by~\eqref{eq:qorder}.
So is its expectation under the distribution of $\pos$, which does not depend on the selection rule.
\end{proof}

\begin{proof}[Proof of Corollary~\ref{cor:pairs}]
Let $\varphi(i, i') \coloneqq f(x_t^i, x_{t+1}^{i'})$ for $i, i' \in [N]$, so that
\begin{align*}
  f(x_t, x_{t+1}) f(\tilde{x}_t, \tilde{x}_{t+1}) = \varphi(\refidx_t, \refidx_{t+1}) \varphi(\outidx_t, \outidx_{t+1}).
\end{align*}
By Lemma~\ref{lem:pair}, its conditional expectation given $\pos$, $\refidx_{t+1} = j$, and $\outidx_{t+1} = j'$ is~\eqref{eq:pairs}.
With $\chi(j) \coloneqq \sum_i \bsprob_t(i \mid j) \varphi(i, j)$, this is $\chi(j) \chi(j')$ for \BS{} and \FMBS{}, and also for \AT{} and \MBS{} if $j \neq j'$.
For $j = j'$, it is $\sum_i \bsprob_t(i \mid j) \varphi(i, j)^2 \geq \chi(j)^2$ for \AT{}, and at most $\chi(j)^2$ for \MBS{}, by the third part of Lemma~\ref{lem:forced}.
Hence the conditional expectation of $\varphi(\refidx_t, \refidx_{t+1}) \varphi(\outidx_t, \outidx_{t+1})$ given $\pos$ is $\expprod_{t+1}(\chi)$ for \BS{} and \FMBS{}, at least $\expprod_{t+1}^{\AT{}}(\chi)$ for \AT{}, and at most $\expprod_{t+1}^{\MBS{}}(\chi)$ for \MBS{}, so that it is ordered as claimed, by~\eqref{eq:qorder} at time $t + 1$.
The first claim follows by taking expectations under the distribution of $\pos$, which does not depend on the selection rule.
These expectations are finite, because $2 |f(x_t, x_{t+1}) f(\tilde{x}_t, \tilde{x}_{t+1})|$ is at most $f(x_t, x_{t+1})^2 + f(\tilde{x}_t, \tilde{x}_{t+1})^2$, where $\tilde{x}_{1:T} \sim \target$ by Proposition~\ref{prop:reversible}.
If the particles are pairwise distinct, then $\tilde{x}_t = x_t$ and $\tilde{x}_{t+1} = x_{t+1}$ if and only if $\refidx_t = \outidx_t$ and $\refidx_{t+1} = \outidx_{t+1}$, almost surely, and, by Lemma~\ref{lem:pair},
\begin{align*}
  & \bbP_{\pos}(\refidx_t = \outidx_t, \refidx_{t+1} = \outidx_{t+1}) \\
  & \quad = \sum_{j=1}^N \bbP_{\pos}(\refidx_{t+1} = \outidx_{t+1} = j) \sum_{i=1}^N \bsprob_t(i \mid j) \kernel_{t,j,j}(i, i).
\end{align*}
By~\eqref{eq:qorder} at time $t + 1$ with $\varphi = \mathbf{1}_{\{j\}}$, the first factor is nondecreasing from \MBS{} to \FMBS{}, \BS{}, and \AT{}, and at most as large under \CMBS{} as under \FMBS{}.
So is the second factor, which is the same under \MBS{} and \CMBS{}, because $\forced_{\mathbf{p}}(i, i) = 1 - \sum_{i' \neq i} \forced_{\mathbf{p}}(i, i') \leq p_i$ by the second part of Lemma~\ref{lem:forced}.
Since both factors are nonnegative, the second claim follows by taking expectations.
\end{proof}

\begin{proof}[Proof of Proposition~\ref{prop:washout}]
For $i \in [N]$, let $\psi_i \coloneqq \bsprob_{T-1} \dotsm \bsprob_t \mathbf{1}_{\{i\}}$, so that $\psi_i(j)$ is the probability that the backward chain with the kernels $\bsprob_s$, started from $j$ at time $T$, is at $i$ at time $t$, and, for $S \subseteq [N]$, let $\psi_S \coloneqq \sum_{i \in S} \psi_i = \bsprob_{T-1} \dotsm \bsprob_t \mathbf{1}_S$.
As in the proof of Theorem~\ref{thm:ordering}, $\bbP_{\pos}(\tilde{x}_t = x_t) = \sum_S \expprod_t(\mathbf{1}_S)$, where the sum is over the classes $S$ of indices whose particles at time $t$ are equal.
For \BS{} and \FMBS{}, $\expprod_t(\mathbf{1}_S) = \expprod_T(\psi_S)$, by~\eqref{eq:qt}, so that, given $\pos$, the difference of the update rates is
\begin{align*}
  \sum_S \bigl\{\weightvec_T(\psi_S)^2 - \langle \psi_S, \forced_{\weightvec_T} \psi_S \rangle_{\weightvec_T}\bigr\}.
\end{align*}
Each term is minus the lag-one autocovariance of $\psi_S$ for the Markov chain with kernel $\forced_{\weightvec_T}$ started from its invariant distribution $\weightvec_T$.
It is therefore nonnegative, by Lemma~\ref{lem:forced}, and upper bounded by the variance $\weightvec_T(\psi_S^2) - \weightvec_T(\psi_S)^2$ of $\psi_S$ under $\weightvec_T$, by the Cauchy--Schwarz inequality.
Summing over $S$,
\begin{align*}
  & \sum_S \bigl\{\weightvec_T(\psi_S^2) - \weightvec_T(\psi_S)^2\bigr\} \\
  & \quad = \frac{1}{2} \sum_{j, j'} \weight_T^j \weight_T^{j'} \sum_S \{\psi_S(j) - \psi_S(j')\}^2 \\
  & \quad \leq \max_{j, j'} \frac{1}{2} \sum_S |\psi_S(j) - \psi_S(j')| \\
  & \quad \leq \max_{j, j'} {\frac{1}{2} \sum_{i=1}^N} |\psi_i(j) - \psi_i(j')|,
\end{align*}
because $|\psi_S(j) - \psi_S(j')| \leq 1$ and $\sum_{j, j'} \weight_T^j \weight_T^{j'} = 1$, and by the triangle inequality, since $\psi_S = \sum_{i \in S} \psi_i$.
The right-hand side is the maximal total-variation distance between the distributions of the backward chain at time $t$ started from two indices at time $T$, which is at most $\prod_{s=t}^{T-1} \dobrushin_s$~\citep{dobrushin1956central}.
The condition on $\transdens_{s+1}$ is the one-step case, with a constant $\mixconst$, of the mixing condition of \citet[Section~3.5.2]{delmoral2004feynman}.
Under this condition, for all $i, j \in [N]$,
\begin{align*}
  \transdens_{s+1}(x_s^i, x_{s+1}^j) \geq \mixconst \max_k \transdens_{s+1}(x_s^k, x_{s+1}^j),
\end{align*}
so that $\bsprob_s(i \mid j) \geq \mixconst \weight_s^i$.
Hence $\dobrushin_s \leq 1 - \mixconst$.
\end{proof}

\begin{proof}[Proof of Proposition~\ref{prop:gain}]
We give the proof for \MBS{}, that for \CMBS{} being the same with \CMBS{} in place of \MBS{}.
Because $\sum_i \excess_{t,j,j'}(\mathbf{1}_{\{i\}}) = -\gain_t(j, j')$, the sum over $\varphi = \mathbf{1}_{\{i\}}$, $i \in [N]$, of the last term of~\eqref{eq:qt} is $-\bbE_{\MBS{}}\{\gain_t(\refidx_{t+1}, \outidx_{t+1}) \mid \pos\}$.
By~\eqref{eq:qt} and~\eqref{eq:qorder} at time $t + 1$, it follows that
\begin{align*}
  & \sum_{i=1}^N \expprod_t^{\MBS{}}(\mathbf{1}_{\{i\}}) \leq \sum_{i=1}^N \expprod_{t+1}^{\FMBS{}}(\bsprob_t \mathbf{1}_{\{i\}}) \\
  & \quad - \bbE_{\MBS{}}\bigl\{\gain_t(\refidx_{t+1}, \outidx_{t+1}) \bigm| \pos\bigr\},
\end{align*}
and $\sum_i \expprod_{t+1}^{\FMBS{}}(\bsprob_t \mathbf{1}_{\{i\}}) = \sum_i \expprod_t^{\FMBS{}}(\mathbf{1}_{\{i\}})$, by~\eqref{eq:qt}.
The claim follows by taking expectations, because $\bbP_{\pos}(\tilde{x}_t = x_t) = \sum_i \expprod_t(\mathbf{1}_{\{i\}})$ for both rules.
Finally, $\forced_{\mathbf{p}}(i, i) = 1 - \sum_{i' \neq i} \forced_{\mathbf{p}}(i, i') \leq p_i$, by the second part of Lemma~\ref{lem:forced}, and $\coupling_{\mathbf{p},\mathbf{q}}(i, i) = q_i \forced_{\mathbf{p} \odot \mathbf{q}}(i, i) \leq q_i$, so that $\gain_t(j, j') \geq 0$.
\end{proof}

\begin{proof}[Proof of Proposition~\ref{prop:imh}]
Because $\transdens_{t+1}(x_t^i, x_{t+1}^j)$ does not depend on $i$, $\bsprob_t(\cdot \mid j) = \weightvec_t$ for all $j \in [N]$ and $t \in [T-1]$.
The kernels of \CMBS{} are therefore $\kernel_{t,j,j} = \forced_{\weightvec_t}$ and, for $j \neq j'$, $\kernel_{t,j,j'} = \coupling_{\weightvec_t, \weightvec_t}$.
For a probability vector $\mathbf{p}$ on $[2]$, $\coupling_{\mathbf{p},\mathbf{p}} = \forced_{\mathbf{p}}$.
Indeed, let $p_1 \geq p_2$, the other case being symmetric.
By~\eqref{eq:forced}, $\forced_{\mathbf{p}}(1, 2) = p_2 / p_1$ and $\forced_{\mathbf{p}}(2, 1) = 1$, so that, with $\mathbf{u} \coloneqq \mathbf{p} \odot \mathbf{p}$, $\forced_{\mathbf{u}}(1, 2) = (p_2 / p_1)^2$ and $\forced_{\mathbf{u}}(2, 1) = 1$.
By~\eqref{eq:coupling}, $\coupling_{\mathbf{p},\mathbf{p}}(1, 2) = p_2 + p_1 (p_2 / p_1)^2 = p_2 / p_1$ and $\coupling_{\mathbf{p},\mathbf{p}}(2, 1) = p_1 + p_2 = 1$.
By Lemma~\ref{lem:pair}, conditionally on $\pos$, the indices $(\refidx_{1:T}, \outidx_{1:T})$ therefore have the probability
\begin{align*}
  \prod_{t=1}^T \weight_t^{\refidx_t} \forced_{\weightvec_t}(\refidx_t, \outidx_t),
\end{align*}
which proves the claim.
\end{proof}

\section{Averages over the time steps}\label{sec:averages}

Table~\ref{tab:averages} gives the averages over the time steps of the update rates and of the ESJD in the experiment of Section~\ref{sec:numerics}, for each model, selection rule, and number of particles.

\begin{table*}[t]
\caption{Averages over the time steps of the update rates and of the ESJD, for the four selection rules and $N \in \{2, 4, 8, 16, 32, 64\}$, on the linear-Gaussian (LG) model with $T = 200$ and the nonlinear (NL) model with $T = 100$, estimated as in Figure~\ref{fig:ratios}.
The standard errors are below $10^{-4}$ for the update rates and below $0.1\%$ of the ESJD.}
\label{tab:averages}
\centering
\renewcommand{\arraystretch}{1.15}
\begin{tabular}{@{}l cccccc cccccc@{}}
\toprule
 & \multicolumn{6}{c}{update rates} & \multicolumn{6}{c}{ESJD} \\
\cmidrule(lr){2-7} \cmidrule(l){8-13}
$N$ & 2 & 4 & 8 & 16 & 32 & 64 & 2 & 4 & 8 & 16 & 32 & 64 \\
\midrule
\multicolumn{13}{@{}l}{\emph{LG, $d = 1$}} \\
\BS{} & 0.310 & 0.569 & 0.749 & 0.861 & 0.925 & 0.960 & 0.296 & 0.536 & 0.700 & 0.801 & 0.859 & 0.892 \\
\FMBS{} & 0.312 & 0.570 & 0.750 & 0.861 & 0.925 & 0.960 & 0.297 & 0.537 & 0.701 & 0.802 & 0.860 & 0.892 \\
\MBS{} & 0.419 & 0.642 & 0.782 & 0.873 & 0.929 & 0.962 & 0.373 & 0.584 & 0.721 & 0.808 & 0.862 & 0.893 \\
\CMBS{} & 0.500 & 0.789 & 0.919 & 0.969 & 0.988 & 0.995 & 0.419 & 0.667 & 0.796 & 0.860 & 0.893 & 0.910 \\
\addlinespace
\multicolumn{13}{@{}l}{\emph{LG, $d = 3$}} \\
\BS{} & 0.132 & 0.293 & 0.465 & 0.623 & 0.751 & 0.844 & 0.379 & 0.833 & 1.31 & 1.75 & 2.10 & 2.36 \\
\FMBS{} & 0.133 & 0.294 & 0.466 & 0.624 & 0.752 & 0.845 & 0.380 & 0.835 & 1.31 & 1.75 & 2.10 & 2.36 \\
\MBS{} & 0.180 & 0.359 & 0.523 & 0.662 & 0.773 & 0.854 & 0.499 & 0.994 & 1.45 & 1.84 & 2.15 & 2.38 \\
\CMBS{} & 0.189 & 0.394 & 0.588 & 0.746 & 0.856 & 0.923 & 0.518 & 1.07 & 1.58 & 2.00 & 2.31 & 2.51 \\
\addlinespace
\multicolumn{13}{@{}l}{\emph{NL, $\sigma_W^2 = 10$}} \\
\BS{} & 0.274 & 0.504 & 0.685 & 0.815 & 0.896 & 0.944 & 3.19 & 7.26 & 12.1 & 16.5 & 19.5 & 21.3 \\
\FMBS{} & 0.277 & 0.506 & 0.687 & 0.816 & 0.897 & 0.944 & 3.25 & 7.32 & 12.2 & 16.5 & 19.5 & 21.3 \\
\MBS{} & 0.379 & 0.585 & 0.732 & 0.836 & 0.905 & 0.947 & 4.12 & 8.16 & 12.8 & 16.9 & 19.7 & 21.4 \\
\CMBS{} & 0.450 & 0.714 & 0.868 & 0.945 & 0.978 & 0.991 & 4.66 & 9.29 & 14.3 & 18.2 & 20.6 & 22.0 \\
\addlinespace
\multicolumn{13}{@{}l}{\emph{NL, $\sigma_W^2 = 1$}} \\
\BS{} & 0.163 & 0.346 & 0.527 & 0.682 & 0.801 & 0.883 & 0.751 & 1.63 & 2.56 & 3.47 & 4.24 & 4.81 \\
\FMBS{} & 0.164 & 0.348 & 0.528 & 0.683 & 0.802 & 0.883 & 0.769 & 1.65 & 2.59 & 3.51 & 4.28 & 4.83 \\
\MBS{} & 0.236 & 0.435 & 0.596 & 0.724 & 0.822 & 0.892 & 1.00 & 1.93 & 2.83 & 3.66 & 4.37 & 4.86 \\
\CMBS{} & 0.261 & 0.509 & 0.710 & 0.847 & 0.927 & 0.968 & 1.07 & 2.12 & 3.10 & 3.96 & 4.63 & 5.08 \\
\bottomrule
\end{tabular}
\end{table*}

\section{Potential functions of the previous state and other resampling schemes}\label{sec:previous-state}

Our proofs require multinomial resampling and potential functions which depend on the current state only.
This section removes both requirements with the construction of \citet[Appendix~D]{finke2026resampling}.
In Algorithm~\ref{alg:previous-state}, the ancestor index of each state of the reference is drawn by ancestor sampling~\citep{lindsten2014particle} in the forward pass, which is followed by the backward pass of Algorithm~\ref{alg:csmc}.
Algorithm~\ref{alg:previous-state} is reversible for all the kernels of Algorithm~\ref{alg:csmc} (Proposition~\ref{prop:previous-state}).
The results of Sections~\ref{sec:update-rates} and~\ref{sec:gains} hold for it (Proposition~\ref{prop:orderings}).
With multinomial resampling and potential functions of the current state, it gives the same joint distribution of the reference and the output as Algorithm~\ref{alg:csmc} (Proposition~\ref{prop:reduction}).
Without ancestor sampling, the usual implementations of \AT{}, \FMAT{}, \BS{}, \FMBS{}, and \MBS{} leave the target distribution invariant (Proposition~\ref{prop:ancestors}).

\subsection{Model and resampling schemes}

Let $\pot_1 \colon \statespace \to (0, \infty)$ and $\pot_t \colon \statespace^2 \to (0, \infty)$, $2 \leq t \leq T$, be potential functions, and let the target distribution have density
\begin{equation}
  \target(x_{1:T}) = {\frac{\transdens_1(x_1) \pot_1(x_1)}{\normconst} \prod_{t=2}^T \transdens_t(x_{t-1}, x_t) \pot_t(x_{t-1}, x_t)}, \label{eq:target-previous}
\end{equation}
where $\normconst \in (0, \infty)$.
At each time $t \in [T-1]$, the vector $\anc_t \coloneqq (\anc_t^1, \dotsc, \anc_t^N)$ of the indices of the ancestors at time $t$ of the particles at time $t + 1$ is drawn from a resampling scheme $\resdist_t$ (Definition~\ref{def:resampling}).
The normalized weights at the times $t \geq 2$ are
\begin{align*}
  \weight_t^n \coloneqq \frac{\pot_t(x_{t-1}^{\anc_{t-1}^n}, x_t^n)}{\sum_{i=1}^N \pot_t(x_{t-1}^{\anc_{t-1}^i}, x_t^i)}, \quad n \in [N],
\end{align*}
which depend on the ancestor indices at time $t - 1$.

\begin{definition}\label{def:resampling}
A resampling scheme at time $t \in [T-1]$ is a distribution $\resdist_t$ on $[N]^N$, which may depend on the particles at time $t$ and on their weights, such that, for every $i \in [N]$, the expected number of indices $n \in [N]$ with $\anc_t^n = i$ under $\resdist_t$ is $N \weight_t^i$.
\end{definition}

Multinomial resampling, with $\resdist_t(\anc_t) = \prod_{n=1}^N \weight_t^{\anc_t^n}$, is a resampling scheme, as are stratified, systematic, residual, and killing resampling, among others~\citep[Section~3]{finke2026resampling}.
The first part of Lemma~\ref{lem:extended} extends to this setting (Lemma~\ref{lem:extended-previous} below), whereas the second part fails, because the weights depend on the ancestor indices, and because the ancestor indices need not be independent (Section~\ref{sec:locally-optimal-ancestors} gives an example).

For $t \in [T-1]$, $i \in [N]$, and $x' \in \statespace$, let
\begin{equation}
  \bsprob_t(i \mid x') \coloneqq \frac{\weight_t^i \transdens_{t+1}(x_t^i, x') \pot_{t+1}(x_t^i, x')}{\sum_{k=1}^N \weight_t^k \transdens_{t+1}(x_t^k, x') \pot_{t+1}(x_t^k, x')}, \label{eq:bs-previous}
\end{equation}
and let $\bsprob_t(i \mid j) \coloneqq \bsprob_t(i \mid x_{t+1}^j)$ for $j \in [N]$.
These are the backward-sampling probabilities of \citet{lindsten2013backward}, which reduce to~\eqref{eq:bs} when $\pot_{t+1}$ depends on the current state only.
Below, the denominator is almost surely positive for $x' = x_{t+1}^j$, $j \in [N]$, because the weights are positive and $x_{t+1}^j$ is either drawn from $\transdens_{t+1}(x_t^{\anc_t^j}, \cdot)$ or the state $x_{t+1}$ of the reference, whose density given its predecessor $x_t$, a particle at time $t$, is positive.
The kernels $\kernel_T$ and $\kernel_{t,j,j'}$ of the six selection rules are defined as in Section~\ref{sec:method}, with these probabilities.
For the six rules, $\kernel_{t,j,j}$ is the identity kernel, the Gibbs kernel $\kernel_{t,j,j}(i, \cdot) = \bsprob_t(\cdot \mid j)$, or the forced move $\forced_{\bsprob_t(\cdot \mid j)}$, which depend on $j$ only through $\bsprob_t(\cdot \mid j)$.

\subsection{Ancestor sampling of the reference}

\begin{algorithm}[CSMC with ancestor sampling and a general backward pass]\label{alg:previous-state}
Given the reference trajectory $x_{1:T}$ and kernels $\kernel_T$ and $\kernel_{t,j,j'}$ as in Algorithm~\ref{alg:csmc}, with the probabilities~\eqref{eq:bs-previous} and functions of the particle positions and of the ancestor indices:
\begin{algsteps}
  \item\label{alg:previous-state:initial} Draw $\refidx_1$ uniformly on $[N]$, set $x_1^{\refidx_1} \coloneqq x_1$, and draw $x_1^n \sim \transdens_1$ for $n \neq \refidx_1$.
  \item\label{alg:previous-state:forward} For $t = 1, \dotsc, T - 1$, draw $c \sim \bsprob_t(\cdot \mid x_{t+1})$, and then $(\refidx_{t+1}, \anc_t)$ with probabilities proportional to $\resdist_t(\anc_t)$ on the set where $\anc_t^{\refidx_{t+1}} = c$.
  Set $x_{t+1}^{\refidx_{t+1}} \coloneqq x_{t+1}$ and draw $x_{t+1}^n \sim \transdens_{t+1}(x_t^{\anc_t^n}, \cdot)$ for $n \neq \refidx_{t+1}$.
  \item\label{alg:previous-state:backward} Run steps~\ref{alg:csmc:final} to~\ref{alg:csmc:return} of Algorithm~\ref{alg:csmc}.
\end{algsteps}
\end{algorithm}

In \algstep{alg:previous-state:forward}, the index $c$ of the ancestor of the state $x_{t+1}$ of the reference is drawn by ancestor sampling~\citep{lindsten2014particle}.
The pair $(\refidx_{t+1}, \anc_t)$ is then drawn by the conditional resampling of \citet[Algorithm~2]{finke2026resampling}, given that $c$ is the ancestor of the particle $\refidx_{t+1}$.
By Definition~\ref{def:resampling}, the normalizing constant of its probabilities is $N \weight_t^c$.
Equivalently, $\refidx_{t+1} = n$ with probability proportional to the probability that $\anc_t^n = c$ under $\resdist_t$.
The other ancestor indices are then drawn from their conditional distribution under $\resdist_t$ given $\anc_t^{\refidx_{t+1}} = c$.
\citet[Section~3]{finke2026resampling} give algorithms for these draws for the resampling schemes above.
With multinomial resampling, $\refidx_{t+1}$ is uniform on $[N]$, while the other ancestor indices are independent draws from $\weightvec_t$, as in \algstep{alg:csmc:forward}.
Unlike in Algorithm~\ref{alg:csmc}, the indices of the reference cannot be fixed in general, because the resampling schemes need not be exchangeable~\citep[Section~2.3]{finke2026resampling}.
\algstep{alg:previous-state:forward} computes the probabilities $\bsprob_t(\cdot \mid x_{t+1})$, at a cost of order $N$ per time step, so that the costs of the six selection rules are of order $NT$, as in Section~\ref{sec:method}.

For each of the six selection rules, we also consider \algstep{alg:previous-state:forward} with $c$ drawn from the kernel $\kernel_{t,j,j}$ of the rule, with $\bsprob_t(\cdot \mid x_{t+1})$ in place of $\bsprob_t(\cdot \mid j)$, started from $\refidx_t$.
With the Gibbs kernels of \BS{} and \FMBS{}, this is \algstep{alg:previous-state:forward} itself.
With the identity kernels of \AT{} and \FMAT{}, $c = \refidx_t$, so that the ancestor of each state of the reference is its predecessor, as in \algstep{alg:csmc:forward}.
With the forced moves of \MBS{} and \CMBS{}, $c$ is drawn by a forced move from $\refidx_t$.

\begin{lemma}\label{lem:extended-previous}
Let $x_{1:T} \sim \target$ and run steps~\ref{alg:previous-state:initial} and~\ref{alg:previous-state:forward} of Algorithm~\ref{alg:previous-state}, with $c$ drawn from the kernels $\kernel_{t,j,j}$ of one of the six selection rules.
Then $x_{1:T} = (x_1^{\refidx_1}, \dotsc, x_T^{\refidx_T})$ and $(\pos, \anc, \refidx_{1:T})$ has the density
\begin{equation}
\begin{split}
  & \frac{\weight_T^{\refidx_T}}{\normconst N^T} \Bigl\{\sum_{i=1}^N \pot_1(x_1^i)\Bigr\} \prod_{n=1}^N \transdens_1(x_1^n) \prod_{t=1}^{T-1} \resdist_t(\anc_t) \\
  & \times \prod_{t=2}^T \Bigl\{\sum_{i=1}^N \pot_t(x_{t-1}^{\anc_{t-1}^i}, x_t^i)\Bigr\} \prod_{n=1}^N \transdens_t(x_{t-1}^{\anc_{t-1}^n}, x_t^n) \\
  & \times \prod_{t=1}^{T-1} \kernel_{t,\refidx_{t+1},\refidx_{t+1}}(\anc_t^{\refidx_{t+1}}, \refidx_t).
\end{split} \label{eq:extended-previous}
\end{equation}
\end{lemma}

\begin{proof}
For $t \in [T-1]$, let $c_t \coloneqq \anc_t^{\refidx_{t+1}}$, the index $c$ drawn in \algstep{alg:previous-state:forward} at time $t$.
Because $x_{t+1}^{\refidx_{t+1}} = x_{t+1}$, it is drawn from $\kernel_{t,\refidx_{t+1},\refidx_{t+1}}(\refidx_t, \cdot)$.
Steps~\ref{alg:previous-state:initial} and~\ref{alg:previous-state:forward} of Algorithm~\ref{alg:previous-state} therefore draw $\refidx_1$ with probability $1 / N$, and, at each time $t \in [T-1]$, $c_t$ with probability $\kernel_{t,\refidx_{t+1},\refidx_{t+1}}(\refidx_t, c_t)$ and $(\refidx_{t+1}, \anc_t)$ with probability $\resdist_t(\anc_t) / (N \weight_t^{c_t})$.
With the density~\eqref{eq:target-previous} of $x_{1:T} = (x_1^{\refidx_1}, \dotsc, x_T^{\refidx_T})$ and the densities of the other particles, $(\pos, \anc, \refidx_{1:T})$ has the density
\begin{align*}
  & \target(x_1^{\refidx_1}, \dotsc, x_T^{\refidx_T}) \frac{1}{N} \prod_{n \neq \refidx_1} \transdens_1(x_1^n) \prod_{t=1}^{T-1} \kernel_{t,\refidx_{t+1},\refidx_{t+1}}(\refidx_t, c_t) \\
  & \times \frac{\resdist_t(\anc_t)}{N \weight_t^{c_t}} \prod_{n \neq \refidx_{t+1}} \transdens_{t+1}(x_t^{\anc_t^n}, x_{t+1}^n).
\end{align*}
Let $t \in [T-1]$ and $x' \coloneqq x_{t+1}^{\refidx_{t+1}}$.
By~\eqref{eq:bs-previous}, for $i \in [N]$,
\begin{align*}
  & \weight_t^i \transdens_{t+1}(x_t^i, x') \pot_{t+1}(x_t^i, x') \\
  & \quad = \bsprob_t(i \mid \refidx_{t+1}) \sum_{k=1}^N \weight_t^k \transdens_{t+1}(x_t^k, x') \pot_{t+1}(x_t^k, x').
\end{align*}
Because $\kernel_{t,\refidx_{t+1},\refidx_{t+1}}$ is reversible with respect to $\bsprob_t(\cdot \mid \refidx_{t+1})$, it follows that
\begin{align*}
  & \weight_t^{\refidx_t} \transdens_{t+1}(x_t^{\refidx_t}, x') \pot_{t+1}(x_t^{\refidx_t}, x') \kernel_{t,\refidx_{t+1},\refidx_{t+1}}(\refidx_t, c_t) \\
  & \quad = \weight_t^{c_t} \transdens_{t+1}(x_t^{c_t}, x') \pot_{t+1}(x_t^{c_t}, x') \kernel_{t,\refidx_{t+1},\refidx_{t+1}}(c_t, \refidx_t).
\end{align*}
The left-hand side is $\weight_t^{\refidx_t}$ times the factor of~\eqref{eq:target-previous} at time $t + 1$ times the factor $\kernel_{t,\refidx_{t+1},\refidx_{t+1}}(\refidx_t, c_t)$ of the density above.
On the right-hand side, $\transdens_{t+1}(x_t^{c_t}, x')$ completes the product over $n \neq \refidx_{t+1}$, because $c_t = \anc_t^{\refidx_{t+1}}$.
The factor $\weight_t^{c_t}$ cancels with the denominator of $\resdist_t(\anc_t) / (N \weight_t^{c_t})$.
Similarly, $\transdens_1(x_1^{\refidx_1})$ completes the product over $n \neq \refidx_1$.
The density of $(\pos, \anc, \refidx_{1:T})$ is therefore
\begin{align*}
  & \frac{\pot_1(x_1^{\refidx_1})}{\normconst N^T} \prod_{n=1}^N \transdens_1(x_1^n) \prod_{t=1}^{T-1} \frac{\pot_{t+1}(x_t^{c_t}, x_{t+1}^{\refidx_{t+1}})}{\weight_t^{\refidx_t}} \resdist_t(\anc_t) \\
  & \times \prod_{n=1}^N \transdens_{t+1}(x_t^{\anc_t^n}, x_{t+1}^n) \kernel_{t,\refidx_{t+1},\refidx_{t+1}}(c_t, \refidx_t).
\end{align*}
By the definition of the weights, and because $c_t = \anc_t^{\refidx_{t+1}}$,
\begin{align*}
  \pot_1(x_1^{\refidx_1}) & = \weight_1^{\refidx_1} \sum_{i=1}^N \pot_1(x_1^i), \\
  \pot_{t+1}(x_t^{c_t}, x_{t+1}^{\refidx_{t+1}}) & = \weight_{t+1}^{\refidx_{t+1}} \sum_{i=1}^N \pot_{t+1}(x_t^{\anc_t^i}, x_{t+1}^i).
\end{align*}
The product of these potentials divided by $\weight_1^{\refidx_1} \dotsm \weight_{T-1}^{\refidx_{T-1}}$ is therefore $\weight_T^{\refidx_T}$ times the product of the sums, which gives~\eqref{eq:extended-previous}.
\end{proof}

The density~\eqref{eq:extended-previous} is the product of a function of $(\pos, \anc)$ and of
\begin{align*}
  \weight_T^{\refidx_T} \prod_{t=1}^{T-1} \kernel_{t,\refidx_{t+1},\refidx_{t+1}}(\anc_t^{\refidx_{t+1}}, \refidx_t),
\end{align*}
which is a probability distribution of $\refidx_{1:T}$.
Hence the marginal density of $(\pos, \anc)$ does not depend on the selection rule.
Conditionally on $(\pos, \anc)$, $\refidx_T \sim \weightvec_T$ and, for $t = T - 1, \dotsc, 1$, given also $\refidx_{t+1:T}$, $\refidx_t \sim \kernel_{t,\refidx_{t+1},\refidx_{t+1}}(\anc_t^{\refidx_{t+1}}, \cdot)$.
Conversely, by Lemma~\ref{lem:extended-previous}, if $(\pos, \anc, \refidx_{1:T})$ has the density~\eqref{eq:extended-previous}, then $(x_1^{\refidx_1}, \dotsc, x_T^{\refidx_T}) \sim \target$.
With multinomial resampling, potential functions of the current state, and the identity kernels, \eqref{eq:extended-previous} is the extended distribution~\eqref{eq:extended}, with $\refidx_t = \anc_t^{\refidx_{t+1}}$ for $t \in [T-1]$, as in the first part of Lemma~\ref{lem:extended}.
With the identity and the Gibbs kernels, it is the extended distribution of \citet[Lemma~1]{finke2026resampling} for CSMC with ancestral tracing and with ancestor sampling, respectively.

\begin{proposition}\label{prop:previous-state}
For any kernels $\kernel_T$ and $\kernel_{t,j,j'}$ as in Algorithm~\ref{alg:previous-state}, the Markov kernel defined by Algorithm~\ref{alg:previous-state} is reversible with respect to $\target$.
\end{proposition}

\begin{proof}
Let $x_{1:T} \sim \target$ and run Algorithm~\ref{alg:previous-state}.
With the Gibbs kernels, $\kernel_{t,\refidx_{t+1},\refidx_{t+1}}(\anc_t^{\refidx_{t+1}}, \refidx_t) = \bsprob_t(\refidx_t \mid \refidx_{t+1})$.
By Lemma~\ref{lem:extended-previous}, conditionally on $(\pos, \anc)$, $\refidx_T \sim \weightvec_T$ and, for $t = T - 1, \dotsc, 1$, given also $\refidx_{t+1:T}$, $\refidx_t \sim \bsprob_t(\cdot \mid \refidx_{t+1})$.
Let $t \in [T-1]$.
The kernels of steps~\ref{alg:csmc:final} and~\ref{alg:csmc:backward} of Algorithm~\ref{alg:csmc} depend on $(\pos, \anc)$ only, so that, at the times $s > t$, they draw $\outidx_s$ given $(\pos, \anc)$, $\refidx_{s:T}$, and $\outidx_{s+1:T}$.
Conditionally on $(\pos, \anc)$, $\refidx_{t+1:T}$, and $\outidx_{t+1:T}$, therefore, $\refidx_t \sim \bsprob_t(\cdot \mid \refidx_{t+1})$ and, given $\refidx_t$, $\outidx_t \sim \kernel_{t,\refidx_{t+1},\outidx_{t+1}}(\refidx_t, \cdot)$.
These are the first two sentences of Lemma~\ref{lem:pair}, conditionally on $(\pos, \anc)$ instead of $\pos$.
The proof of Proposition~\ref{prop:reversible} then applies with $(\pos, \anc)$ in place of $\pos$.
\end{proof}

\begin{proposition}\label{prop:orderings}
Let $\tilde{x}_{1:T}$ be the output of Algorithm~\ref{alg:previous-state} started from $x_{1:T} \sim \target$.
Theorem~\ref{thm:ordering}, Corollary~\ref{cor:pairs}, and Propositions~\ref{prop:washout} and~\ref{prop:gain} hold with the probabilities~\eqref{eq:bs-previous}, with the distribution~\eqref{eq:extended-previous} in place of $\exttarget$, and, in Proposition~\ref{prop:washout}, with $\transdens_{s+1}(x, \cdot) \pot_{s+1}(x, \cdot) \geq \mixconst \transdens_{s+1}(x', \cdot) \pot_{s+1}(x', \cdot)$ in place of the condition on $\transdens_{s+1}$.
\end{proposition}

\begin{proof}
Condition on $(\pos, \anc)$ instead of $\pos$ in Sections~\ref{sec:update-rates} and~\ref{sec:gains} and in the proofs of these results in Section~\ref{sec:proofs}, with $\bbP_{\pos,\anc}$ and $\bbE_{\pos,\anc}$ in place of $\bbP_{\pos}$ and $\bbE_{\pos}$.
The proofs use the first two sentences of Lemma~\ref{lem:pair}, which hold conditionally on $(\pos, \anc)$ by the proof of Proposition~\ref{prop:previous-state}.
They also use Proposition~\ref{prop:reversible}, which is replaced by Proposition~\ref{prop:previous-state}, and the fact that the distribution of the conditioning variables does not depend on the selection rule.
This holds for $(\pos, \anc)$, because steps~\ref{alg:previous-state:initial} and~\ref{alg:previous-state:forward} of Algorithm~\ref{alg:previous-state} do not depend on the kernels.
In the proof of Proposition~\ref{prop:washout}, the condition on $\transdens_{s+1} \pot_{s+1}$ gives $\bsprob_s(i \mid j) \geq \mixconst \weight_s^i$ by~\eqref{eq:bs-previous}, as before.
\end{proof}

\begin{proposition}\label{prop:reduction}
Assume multinomial resampling, and that $\pot_t(x_{t-1}, x_t)$ does not depend on $x_{t-1}$ for $2 \leq t \leq T$.
Then, for $x_{1:T} \sim \target$ and any kernels as in Algorithm~\ref{alg:csmc}, the reference and the output have the same joint distribution under Algorithm~\ref{alg:previous-state} as under Algorithm~\ref{alg:csmc}.
\end{proposition}

\begin{proof}
Under these assumptions, $\weightvec_T$ and $\bsprob_t$ are functions of $\pos$, and so are the kernels.
With multinomial resampling and the identity kernels, steps~\ref{alg:previous-state:initial} and~\ref{alg:previous-state:forward} of Algorithm~\ref{alg:previous-state} draw $\refidx_{t+1}$ uniformly on $[N]$, independently of the other variables, so that they generate $(\pos, \anc, \refidx_{1:T})$ with the same distribution as steps~\ref{alg:csmc:indices} and~\ref{alg:csmc:forward} of Algorithm~\ref{alg:csmc}.
By Lemma~\ref{lem:extended-previous}, the marginal density of $(\pos, \anc)$ is therefore the same under both algorithms.
Under Algorithm~\ref{alg:previous-state}, conditionally on $(\pos, \anc)$, and hence on $\pos$, $\refidx_T \sim \weightvec_T$ and $\refidx_t \sim \bsprob_t(\cdot \mid \refidx_{t+1})$ given $\refidx_{t+1:T}$, by Lemma~\ref{lem:extended-previous}.
Under Algorithm~\ref{alg:csmc}, the same holds conditionally on $\pos$, by Lemma~\ref{lem:pair}.
The variables $(\pos, \refidx_{1:T})$ thus have the same distribution under both algorithms.
So do $(\pos, \refidx_{1:T}, \outidx_{1:T})$, because both algorithms draw $\outidx_{1:T}$ given $\pos$ and $\refidx_{1:T}$ with the same kernels.
The reference $(x_1^{\refidx_1}, \dotsc, x_T^{\refidx_T})$ and the output $(x_1^{\outidx_1}, \dotsc, x_T^{\outidx_T})$ are functions of these variables.
\end{proof}

\subsection{Backward passes from the ancestors}

\begin{algorithm}[CSMC with a backward pass from the ancestors]\label{alg:ancestors}
Given the reference trajectory $x_{1:T}$ and the kernels $\kernel_T$ and $\kernel_{t,j,j}$ of \AT{}, \FMAT{}, \BS{}, \FMBS{}, or \MBS{}, with the probabilities~\eqref{eq:bs-previous}:
\begin{algsteps}
  \item\label{alg:ancestors:forward} Run steps~\ref{alg:previous-state:initial} and~\ref{alg:previous-state:forward} of Algorithm~\ref{alg:previous-state}, with $c \coloneqq \refidx_t$ in \algstep{alg:previous-state:forward}.
  \item\label{alg:ancestors:final} Draw $\outidx_T \sim \kernel_T(\refidx_T, \cdot)$.
  \item\label{alg:ancestors:backward} For $t = T - 1, \dotsc, 1$, draw $\outidx_t$ from the kernel $\kernel_{t,\outidx_{t+1},\outidx_{t+1}}$ started from the ancestor $\anc_t^{\outidx_{t+1}}$ of the particle $\outidx_{t+1}$.
  \item\label{alg:ancestors:return} Return $\tilde{x}_{1:T} \coloneqq (x_1^{\outidx_1}, \dotsc, x_T^{\outidx_T})$.
\end{algsteps}
\end{algorithm}

\algstep{alg:ancestors:forward} is the forward pass of CSMC with the resampling schemes $\resdist_t$~\citep[Algorithm~27]{finke2026resampling}.
With multinomial resampling and potential functions of the current state, Algorithm~\ref{alg:ancestors} is the usual implementation of Section~\ref{sec:method}, with the forward pass of Algorithm~\ref{alg:csmc}.

\begin{proposition}\label{prop:ancestors}
The Markov kernel defined by Algorithm~\ref{alg:ancestors} leaves $\target$ invariant.
With the kernels of \AT{} and \FMAT{}, it is reversible with respect to $\target$.
\end{proposition}

\begin{proof}
Let $x_{1:T} \sim \target$ and run Algorithm~\ref{alg:ancestors}.
By Lemma~\ref{lem:extended-previous} with the identity kernels, $\refidx_T \sim \weightvec_T$ conditionally on $(\pos, \anc)$.
Because $\kernel_T$ leaves $\weightvec_T$ invariant, $\outidx_T \sim \weightvec_T$ conditionally on $(\pos, \anc)$.
Because the marginal density of $(\pos, \anc)$ in~\eqref{eq:extended-previous} does not depend on the kernels, \algstep{alg:ancestors:backward} then gives $(\pos, \anc, \outidx_{1:T})$ the density~\eqref{eq:extended-previous}, with $\outidx_{1:T}$ in place of $\refidx_{1:T}$ and the kernels $\kernel_{t,j,j}$ of the selection rule.
By Lemma~\ref{lem:extended-previous}, $\tilde{x}_{1:T} \sim \target$.
With the identity kernels of \AT{} and \FMAT{}, $\outidx_t = \anc_t^{\outidx_{t+1}}$ and $\refidx_t = \anc_t^{\refidx_{t+1}}$ for $t \in [T-1]$.
Conditionally on $(\pos, \anc)$, $(\refidx_{1:T}, \outidx_{1:T})$ therefore has the probability $\weight_T^{\refidx_T} \kernel_T(\refidx_T, \outidx_T)$ if these identities hold, and zero otherwise.
By the reversibility of $\kernel_T$, this probability does not change when $\refidx_{1:T}$ and $\outidx_{1:T}$ are exchanged.
The reversibility then follows as in the proof of Proposition~\ref{prop:reversible}.
\end{proof}

The proof uses \algstep{alg:ancestors:forward} only through $\refidx_T \sim \weightvec_T$ conditionally on $(\pos, \anc)$, which holds for all the kernels of Lemma~\ref{lem:extended-previous}, so that the invariance also holds with ancestor sampling in \algstep{alg:ancestors:forward}.
By Lemma~\ref{lem:extended-previous}, the backward pass of Algorithm~\ref{alg:ancestors} and the forward pass of \algstep{alg:previous-state:forward} with the same kernels $\kernel_{t,j,j}$ both lead to the distribution~\eqref{eq:extended-previous}.
This is the duality of backward and ancestor sampling of \citet[Proposition~3.27]{finke2015extended}, here with forced moves.
\CMBS{} does not have such an implementation, because its kernels for $j \neq j'$ start from the index $\refidx_t$ of the reference.
Under Algorithm~\ref{alg:previous-state}, $\refidx_t \sim \bsprob_t(\cdot \mid \refidx_{t+1})$ conditionally on $(\pos, \anc)$ and $\refidx_{t+1:T}$, by Lemma~\ref{lem:extended-previous}, whereas, with $c \coloneqq \refidx_t$, $\refidx_t = \anc_t^{\refidx_{t+1}}$ is a function of $\anc$ and $\refidx_{t+1}$.

\section{Locally optimal proposals}\label{sec:locally-optimal}

This section reports the experiment of Section~\ref{sec:numerics} on the linear-Gaussian (LG) model with the particle filter which uses the locally optimal proposal~\citep{doucet2000sequential}, whose potential functions depend on the previous state.

\subsection{Particle filter with the locally optimal proposal}

The LG model is $x_1 \sim \mathcal{N}(0, \sigma^2 (1 - \rho^2)^{-1} I_d)$, $x_t \mid x_{t-1} \sim \mathcal{N}(\rho x_{t-1}, \sigma^2 I_d)$, and $y_t \mid x_t \sim \mathcal{N}(x_t, \sigma^2 I_d)$.
The particle filter with the locally optimal proposal is that of Section~\ref{sec:previous-state}, with multinomial resampling, for the model~\eqref{eq:target-previous} in which $\transdens_t(x_{t-1}, \cdot)$, $t \geq 2$, is the density of the conditional distribution of $x_t$ given $x_{t-1}$ and $y_t$,
\begin{align*}
  \mathcal{N}\Bigl(\frac{\rho x_{t-1} + y_t}{2}, \frac{\sigma^2}{2} I_d\Bigr),
\end{align*}
and $\pot_t(x_{t-1}, x_t)$ is the density of $y_t$ given $x_{t-1}$, which is the density of $\mathcal{N}(\rho x_{t-1}, 2 \sigma^2 I_d)$ at $y_t$, whereas $\transdens_1$ is the density of the conditional distribution of $x_1$ given $y_1$, $\mathcal{N}(y_1 / (2 - \rho^2), \sigma^2 (2 - \rho^2)^{-1} I_d)$, and $\pot_1$ is constant.
Because $\transdens_t(x_{t-1}, x_t) \pot_t(x_{t-1}, x_t)$ is the transition density of the LG model times the density of $y_t$ given $x_t$, while $\transdens_1(x_1) \pot_1(x_1)$ is proportional to the density of $x_1$ times that of $y_1$ given $x_1$, the target distribution~\eqref{eq:target-previous} is the smoothing distribution.
The backward-sampling probability $\bsprob_t(i \mid j)$ is therefore proportional to $\weight_t^i$ times the density of $\mathcal{N}(\rho x_t^i, \sigma^2 I_d)$ at $x_{t+1}^j$, as for the bootstrap filter.
The four selection rules use these probabilities with Algorithm~\ref{alg:previous-state}, whose forward pass draws the ancestor index of each state of the reference by ancestor sampling.
Their kernels start from the indices of the reference, as in \algstep{alg:csmc:backward}, because the ancestor $\anc_t^{\refidx_{t+1}}$ drawn by ancestor sampling need not be $\refidx_t$.
With multinomial resampling, the indices of the reference may be fixed, as in the experiment of Section~\ref{sec:numerics}.

\subsection{Dependence of the ancestor indices on the proposal}\label{sec:locally-optimal-ancestors}

The backward-sampling probabilities do not involve the proposal, but the conditional distribution of the ancestor indices given the particle positions does.
With $T = 2$, for example, the distribution~\eqref{eq:extended-previous} with the identity kernels shows that, given the particle positions and the final index $\refidx_2$, the ancestor indices $\anc_1^n$, $n \neq \refidx_2$, are independent, with probabilities proportional to $\transdens_2(x_1^i, x_2^n)$, $i \in [N]$, the density of the proposal of $x_2^n$ from $x_1^i$.
The backward-sampling probability $\bsprob_1(i \mid n)$ is instead proportional to the density of $\mathcal{N}(\rho x_1^i, \sigma^2 I_d)$ at $x_2^n$.
These probabilities differ unless the density of $y_2$ given $x_1 = x_1^i$ does not depend on $i$.
Moreover, the final index $\refidx_2$ given the particle positions does not have the distribution $\weightvec_2$, which depends on the ancestor indices.
The second part of Lemma~\ref{lem:extended} therefore fails, and so does Lemma~\ref{lem:pair} conditionally on the particle positions only.
Under Algorithm~\ref{alg:previous-state}, the first two sentences of Lemma~\ref{lem:pair} hold instead conditionally on the particle positions and the ancestor indices (proof of Proposition~\ref{prop:previous-state}).
Algorithm~\ref{alg:previous-state} is therefore reversible, and the four selection rules are ordered with it as in Theorem~\ref{thm:ordering} and Corollary~\ref{cor:pairs} (Propositions~\ref{prop:previous-state} and~\ref{prop:orderings}), which Section~\ref{sec:locally-optimal-experiment} illustrates.

\subsection{Update rates and ESJD}\label{sec:locally-optimal-experiment}

We repeat the experiment of Section~\ref{sec:numerics} on the LG model, with $\rho = 0.9$, $\sigma = 1$, $T = 200$, $d \in \{1, 3\}$, and $N \in \{2, 4, 8, 16, 32, 64\}$, with $10^6$ data sets drawn from the model for each $d$ and $N$, and with the locally optimal proposal.
For each data set, the forward pass of Algorithm~\ref{alg:previous-state} is run once with the reference $x_{1:T}$, and the four selection rules are applied to the same particles.

Figure~\ref{fig:locally-optimal-ratios} shows the update rates and the ESJD of \FMBS{}, \MBS{}, and \CMBS{} relative to those of \BS{} at each time step, while Table~\ref{tab:locally-optimal-averages} shows their averages over the time steps.
With the locally optimal proposal, the update rates of the four selection rules are larger than with the bootstrap filter of Section~\ref{sec:numerics}, for example those of \BS{} with $N = 2$, increase from $0.310$ to $0.418$ for $d = 1$ and from $0.132$ to $0.302$ for $d = 3$ (Tables~\ref{tab:averages} and~\ref{tab:locally-optimal-averages}).
In agreement with Proposition~\ref{prop:orderings}, \MBS{} and \CMBS{} increase the update rates and the ESJD at every time step in Figure~\ref{fig:locally-optimal-ratios}, whereas, as for the bootstrap filter, the final forced move of \FMBS{} only matters at the last time steps.
On average over the time steps, \MBS{} increases the update rates of \BS{} by $29\%$ for $N = 2$ and by $2\%$ to $4\%$ for $N = 8$, and the ESJD by $23\%$ to $26\%$ for $N = 2$ and by $2\%$ to $4\%$ for $N = 8$, while \FMBS{} increases both by less than $1\%$.
\CMBS{} increases the update rates of \BS{} by $49\%$ to $66\%$ for $N = 2$ and by $17\%$ to $19\%$ for $N = 8$, and the ESJD by $43\%$ to $50\%$ for $N = 2$ and by $12\%$ to $15\%$ for $N = 8$.

\begin{figure*}[t]
\centering
\includegraphics[width=\textwidth]{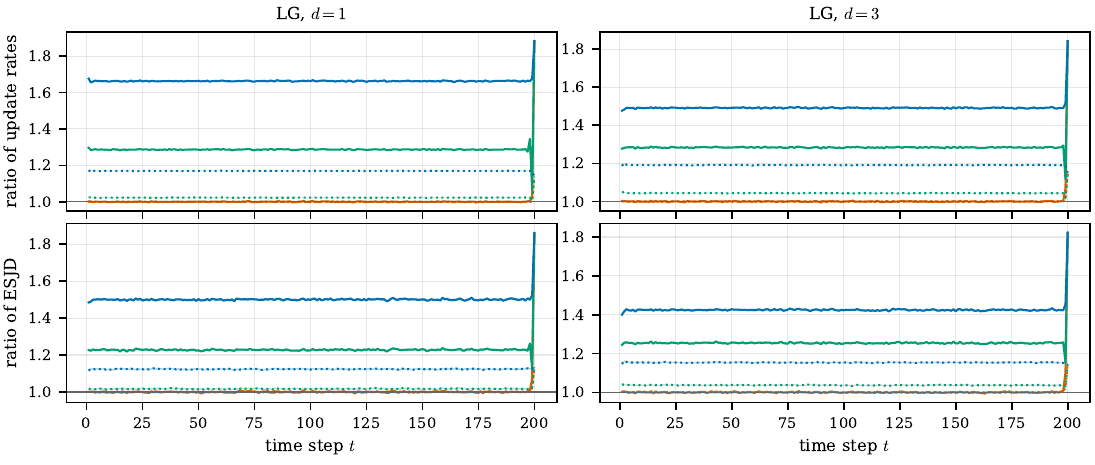}
\caption{Update rates (top) and ESJD (bottom) of forced move plus backward sampling (\FMBS{}), Metropolized backward sampling (\MBS{}), and coupled Metropolized backward sampling (\CMBS{}), relative to those of backward sampling (\BS{}), at each time step $t$, on the linear-Gaussian (LG) model with $T = 200$ and the locally optimal proposal: \FMBS{} with $N = 2$ (\legFMBSNtwo) and $N = 8$ (\legFMBSNeight), \MBS{} with $N = 2$ (\legMBSNtwo) and $N = 8$ (\legMBSNeight), and \CMBS{} with $N = 2$ (\legCMBSNtwo) and $N = 8$ (\legCMBSNeight), the four rules with Algorithm~\ref{alg:previous-state}.
The ratios are those of the averages over $10^6$ data sets drawn from the model, with one iteration for each data set, started from the trajectory drawn with the data.
The shaded bands extend two standard errors on each side of the ratios.}
\label{fig:locally-optimal-ratios}
\end{figure*}

\begin{table*}[t]
\caption{Averages over the time steps of the update rates and of the ESJD, for the four selection rules and $N \in \{2, 4, 8, 16, 32, 64\}$, on the linear-Gaussian (LG) model with $T = 200$ and the locally optimal proposal, estimated as in Figure~\ref{fig:locally-optimal-ratios}.
The standard errors are below $10^{-4}$ for the update rates and below $0.1\%$ of the ESJD.}
\label{tab:locally-optimal-averages}
\centering
\renewcommand{\arraystretch}{1.15}
\begin{tabular}{@{}l cccccc cccccc@{}}
\toprule
 & \multicolumn{6}{c}{update rates} & \multicolumn{6}{c}{ESJD} \\
\cmidrule(lr){2-7} \cmidrule(l){8-13}
$N$ & 2 & 4 & 8 & 16 & 32 & 64 & 2 & 4 & 8 & 16 & 32 & 64 \\
\midrule
\multicolumn{13}{@{}l}{\emph{LG, $d = 1$}} \\
\BS{} & 0.418 & 0.674 & 0.824 & 0.908 & 0.952 & 0.976 & 0.353 & 0.590 & 0.739 & 0.827 & 0.876 & 0.902 \\
\FMBS{} & 0.420 & 0.675 & 0.825 & 0.908 & 0.953 & 0.976 & 0.355 & 0.591 & 0.740 & 0.827 & 0.876 & 0.902 \\
\MBS{} & 0.539 & 0.728 & 0.843 & 0.914 & 0.954 & 0.976 & 0.434 & 0.626 & 0.752 & 0.831 & 0.877 & 0.902 \\
\CMBS{} & 0.696 & 0.894 & 0.964 & 0.988 & 0.996 & 0.999 & 0.530 & 0.731 & 0.830 & 0.880 & 0.905 & 0.917 \\
\addlinespace
\multicolumn{13}{@{}l}{\emph{LG, $d = 3$}} \\
\BS{} & 0.302 & 0.537 & 0.711 & 0.829 & 0.903 & 0.947 & 0.766 & 1.40 & 1.90 & 2.25 & 2.47 & 2.61 \\
\FMBS{} & 0.305 & 0.539 & 0.711 & 0.829 & 0.903 & 0.947 & 0.773 & 1.41 & 1.90 & 2.25 & 2.47 & 2.61 \\
\MBS{} & 0.389 & 0.598 & 0.742 & 0.842 & 0.908 & 0.949 & 0.965 & 1.54 & 1.97 & 2.28 & 2.48 & 2.62 \\
\CMBS{} & 0.452 & 0.700 & 0.846 & 0.927 & 0.967 & 0.986 & 1.09 & 1.75 & 2.19 & 2.45 & 2.61 & 2.69 \\
\bottomrule
\end{tabular}
\end{table*}

\section{Stratified resampling}\label{sec:stratified}

We repeat the experiments of Section~\ref{sec:numerics} and of Section~\ref{sec:locally-optimal} with stratified resampling~\citep{kitagawa1996monte} instead of multinomial resampling, with Algorithm~\ref{alg:previous-state}.
In stratified resampling, $\anc_t^n$ is the index $i$ with $U_n \in (\weight_t^1 + \dotsb + \weight_t^{i-1}, \weight_t^1 + \dotsb + \weight_t^i]$, for independent uniforms $U_n$ on $[(n - 1) / N, n / N)$, $n \in [N]$.
It satisfies Definition~\ref{def:resampling}~\citep[Section~3.1.2]{finke2026resampling}, so that the reversibility and the orderings hold (Propositions~\ref{prop:previous-state} and~\ref{prop:orderings}).
In \algstep{alg:previous-state:forward}, with the conditional stratified resampling of \citet[Algorithm~6]{finke2026resampling}, $\refidx_{t+1} = n$ with probability proportional to the length of the intersection of $[(n - 1) / N, n / N)$ with $(\weight_t^1 + \dotsb + \weight_t^{c-1}, \weight_t^1 + \dotsb + \weight_t^c]$, $\anc_t^{\refidx_{t+1}} = c$, and the other ancestor indices are drawn as in stratified resampling.
Because stratified resampling is not exchangeable, the indices of the reference are not fixed, and $\refidx_1$ is drawn uniformly on $[N]$, as in \algstep{alg:previous-state:initial}.
As in Section~\ref{sec:locally-optimal}, the kernels start from the indices of the reference, and the four selection rules are applied to the same particles, for $10^6$ data sets for each model and $N$.

In Table~\ref{tab:stratified-averages}, the averages over the time steps of \MBS{} and \CMBS{} exceed those of \FMBS{}, and those of \CMBS{} exceed those of \MBS{}, for all the models and numbers of particles.
With $N = 2$, \MBS{} increases the update rates of \BS{} by $35\%$ to $45\%$ with the bootstrap filter and by $28\%$ with the locally optimal proposal, and \CMBS{} by $43\%$ to $64\%$ and by $47\%$ to $63\%$, close to the increases with multinomial resampling.
The update rates of \BS{} are slightly smaller than with multinomial resampling (Tables~\ref{tab:averages} and~\ref{tab:stratified-averages}).

\begin{table*}[t]
\caption{Averages over the time steps of the update rates and of the ESJD with stratified resampling and Algorithm~\ref{alg:previous-state}, for the four selection rules and $N \in \{2, 4, 8, 16, 32, 64\}$, on the linear-Gaussian (LG) model with $T = 200$ and the nonlinear (NL) model with $T = 100$ with the bootstrap filter, and on the LG model with the locally optimal proposal, estimated as in Figure~\ref{fig:ratios}.
The standard errors are below $10^{-4}$ for the update rates and below $0.1\%$ of the ESJD.}
\label{tab:stratified-averages}
\centering
\renewcommand{\arraystretch}{1.15}
\begin{tabular}{@{}l cccccc cccccc@{}}
\toprule
 & \multicolumn{6}{c}{update rates} & \multicolumn{6}{c}{ESJD} \\
\cmidrule(lr){2-7} \cmidrule(l){8-13}
$N$ & 2 & 4 & 8 & 16 & 32 & 64 & 2 & 4 & 8 & 16 & 32 & 64 \\
\midrule
\multicolumn{13}{@{}l}{\emph{LG, $d = 1$}} \\
\BS{} & 0.304 & 0.564 & 0.747 & 0.860 & 0.925 & 0.960 & 0.300 & 0.545 & 0.708 & 0.806 & 0.862 & 0.894 \\
\FMBS{} & 0.306 & 0.565 & 0.748 & 0.860 & 0.925 & 0.960 & 0.301 & 0.545 & 0.709 & 0.807 & 0.863 & 0.894 \\
\MBS{} & 0.411 & 0.638 & 0.780 & 0.872 & 0.929 & 0.961 & 0.379 & 0.595 & 0.729 & 0.814 & 0.865 & 0.895 \\
\CMBS{} & 0.490 & 0.784 & 0.918 & 0.969 & 0.988 & 0.995 & 0.425 & 0.678 & 0.805 & 0.865 & 0.896 & 0.912 \\
\addlinespace
\multicolumn{13}{@{}l}{\emph{LG, $d = 3$}} \\
\BS{} & 0.129 & 0.288 & 0.460 & 0.620 & 0.750 & 0.844 & 0.376 & 0.830 & 1.31 & 1.75 & 2.11 & 2.36 \\
\FMBS{} & 0.130 & 0.289 & 0.461 & 0.620 & 0.750 & 0.844 & 0.377 & 0.832 & 1.31 & 1.76 & 2.11 & 2.36 \\
\MBS{} & 0.176 & 0.353 & 0.518 & 0.659 & 0.771 & 0.854 & 0.496 & 0.994 & 1.45 & 1.85 & 2.16 & 2.39 \\
\CMBS{} & 0.185 & 0.387 & 0.582 & 0.742 & 0.854 & 0.922 & 0.514 & 1.06 & 1.58 & 2.01 & 2.31 & 2.51 \\
\addlinespace
\multicolumn{13}{@{}l}{\emph{NL, $\sigma_W^2 = 10$}} \\
\BS{} & 0.250 & 0.485 & 0.677 & 0.812 & 0.896 & 0.944 & 3.72 & 8.54 & 13.6 & 17.6 & 20.2 & 21.7 \\
\FMBS{} & 0.252 & 0.487 & 0.678 & 0.813 & 0.896 & 0.944 & 3.79 & 8.63 & 13.7 & 17.7 & 20.2 & 21.7 \\
\MBS{} & 0.345 & 0.567 & 0.725 & 0.834 & 0.904 & 0.947 & 4.76 & 9.67 & 14.4 & 18.0 & 20.4 & 21.7 \\
\CMBS{} & 0.409 & 0.691 & 0.861 & 0.943 & 0.977 & 0.991 & 5.34 & 11.0 & 16.1 & 19.5 & 21.4 & 22.3 \\
\addlinespace
\multicolumn{13}{@{}l}{\emph{NL, $\sigma_W^2 = 1$}} \\
\BS{} & 0.153 & 0.334 & 0.518 & 0.677 & 0.799 & 0.882 & 0.759 & 1.67 & 2.63 & 3.55 & 4.31 & 4.85 \\
\FMBS{} & 0.154 & 0.336 & 0.519 & 0.679 & 0.800 & 0.883 & 0.780 & 1.70 & 2.67 & 3.58 & 4.34 & 4.87 \\
\MBS{} & 0.221 & 0.421 & 0.588 & 0.720 & 0.820 & 0.891 & 1.01 & 2.00 & 2.92 & 3.75 & 4.44 & 4.91 \\
\CMBS{} & 0.245 & 0.493 & 0.700 & 0.843 & 0.926 & 0.968 & 1.08 & 2.18 & 3.19 & 4.05 & 4.71 & 5.12 \\
\addlinespace
\multicolumn{13}{@{}l}{\emph{LG, $d = 1$, locally optimal}} \\
\BS{} & 0.402 & 0.666 & 0.821 & 0.907 & 0.952 & 0.976 & 0.362 & 0.602 & 0.748 & 0.832 & 0.879 & 0.904 \\
\FMBS{} & 0.404 & 0.667 & 0.822 & 0.907 & 0.952 & 0.976 & 0.364 & 0.603 & 0.749 & 0.833 & 0.879 & 0.904 \\
\MBS{} & 0.515 & 0.720 & 0.841 & 0.913 & 0.954 & 0.976 & 0.447 & 0.639 & 0.762 & 0.837 & 0.880 & 0.904 \\
\CMBS{} & 0.655 & 0.883 & 0.961 & 0.987 & 0.996 & 0.999 & 0.538 & 0.746 & 0.841 & 0.886 & 0.908 & 0.919 \\
\addlinespace
\multicolumn{13}{@{}l}{\emph{LG, $d = 3$, locally optimal}} \\
\BS{} & 0.281 & 0.522 & 0.702 & 0.825 & 0.901 & 0.946 & 0.749 & 1.40 & 1.90 & 2.26 & 2.48 & 2.62 \\
\FMBS{} & 0.282 & 0.523 & 0.703 & 0.825 & 0.902 & 0.946 & 0.754 & 1.40 & 1.91 & 2.26 & 2.48 & 2.62 \\
\MBS{} & 0.359 & 0.582 & 0.734 & 0.839 & 0.907 & 0.948 & 0.941 & 1.54 & 1.98 & 2.29 & 2.49 & 2.62 \\
\CMBS{} & 0.413 & 0.678 & 0.836 & 0.923 & 0.966 & 0.985 & 1.06 & 1.75 & 2.20 & 2.47 & 2.62 & 2.70 \\
\bottomrule
\end{tabular}
\end{table*}

\section{Differences between the formalization and the paper}\label{sec:formalization}
\subsection{Overview}
The results of the paper are formalized in the Lean proof assistant.
The formalization is available at \url{https://adriencorenflos.github.io/formalization/metropolized-CSMC}.
In order to reduce dependencies, we also provided a formal proof of \citet{liu1996peskun}, \citet[Lemma~3]{tierney1998note}, and \citet{dobrushin1956central} as part of the process.
The formalized proofs follow those of Sections~\ref{sec:proofs} and~\ref{sec:previous-state}, except in the five steps of Sections~\ref{sec:formalization-forced} to~\ref{sec:formalization-reduction}.
Instead of forcing the AI model to follow our proofs, we instead document here where it deviated (despite having been given our original proofs as input) as this may be of interest to some readers.
The statements of the lemmas, propositions, and theorems are themselves faithful to the paper.
Section~\ref{sec:formalization-assumptions} gives the assumptions and conventions of the formalization, and Section~\ref{sec:formalization-claims} the claims of the paper which are not formalized.

\subsection{The forced move}\label{sec:formalization-forced}

In Lemma~\ref{lem:forced}, the third part is derived from the first two without the Gibbs kernel and \citet[Lemma~3]{tierney1998note}.
Because $\forced_{\mathbf{p}}$ is a Markov kernel which leaves $\mathbf{p}$ invariant, the computation of the proof of Lemma~\ref{lem:coupling}, with $\mathbf{q} = \mathbf{p}$ and $\forced_{\mathbf{p}}$ in place of $\coupling_{\mathbf{p},\mathbf{q}}$, gives
\begin{align*}
  & \langle \varphi, \forced_{\mathbf{p}} \varphi \rangle_{\mathbf{p}} - \mathbf{p}(\varphi)^2 \\
  & \quad = -\frac{1}{2} \sum_{i \neq j} p_i \{\forced_{\mathbf{p}}(i, j) - p_j\} \{\varphi(i) - \varphi(j)\}^2,
\end{align*}
which is nonpositive by the second part.

\subsection{The backward Markov chain of the indices}\label{sec:formalization-pair}

In Lemma~\ref{lem:pair}, the conditional distributions are not derived one time step at a time, from the fact that $(\refidx_{t+1:T}, \outidx_{t+1:T})$ is a function of $\refidx_T$, of the ancestor indices at the times $s > t$, and of the draws at these times.
Instead, the probability of $(\refidx_{1:T}, \outidx_{1:T})$ given $\pos$ is computed as the sum over $(\anc, \refidx_T)$ of the probability of $(\anc, \refidx_T)$ given $\pos$, from Lemma~\ref{lem:extended}, times that of $\outidx_{1:T}$ in steps~\ref{alg:csmc:final} and~\ref{alg:csmc:backward} of Algorithm~\ref{alg:csmc}, with $\refidx_{1:T}$ traced back from $\refidx_T$.
The sums over the vectors $\anc_1, \dotsc, \anc_{T-1}$ are taken one at a time.
Each uses that, if $\anc_t^n \sim \bsprob_t(\cdot \mid n)$ independently for $n \in [N]$, then $\anc_t^j \sim \bsprob_t(\cdot \mid j)$, independently of $\anc_t^{j'}$ for $j' \neq j$.
The result is the probability
\begin{align*}
  & \weight_T^{\refidx_T} \kernel_T(\refidx_T, \outidx_T) \\
  & \times \prod_{t=1}^{T-1} \bsprob_t(\refidx_t \mid \refidx_{t+1}) \kernel_{t,\refidx_{t+1},\outidx_{t+1}}(\refidx_t, \outidx_t).
\end{align*}
Its factors are the conditional distributions of Lemma~\ref{lem:pair}.
The same probability is obtained for the usual implementation of \algstep{alg:csmc:backward}, with the case split of our proof on $\outidx_{t+1} = \refidx_{t+1}$ and $\outidx_{t+1} \neq \refidx_{t+1}$.

\subsection{The recursion of the expected products}\label{sec:formalization-recursion}

In Section~\ref{sec:update-rates}, the recursion~\eqref{eq:qt} and the expression~\eqref{eq:pairs} are derived from the joint probabilities given $\pos$, rather than from the expectations conditional on $\pos$, $\refidx_{t+1} = j$, and $\outidx_{t+1} = j'$.
By the probability of $(\refidx_{1:T}, \outidx_{1:T})$ in Section~\ref{sec:formalization-pair},
\begin{align*}
  & \bbP_{\pos}(\refidx_t = i, \outidx_t = i', \refidx_{t+1} = j, \outidx_{t+1} = j') \\
  & \quad = \bbP_{\pos}(\refidx_{t+1} = j, \outidx_{t+1} = j') \bsprob_t(i \mid j) \kernel_{t,j,j'}(i, i').
\end{align*}
Multiplying by $\varphi(i) \varphi(i')$ and summing over $i, i', j, j' \in [N]$ gives~\eqref{eq:qt}.
Multiplying by $\varphi(i, j) \varphi(i', j')$ instead gives the expectation of $f(x_t, x_{t+1}) f(\tilde{x}_t, \tilde{x}_{t+1})$ given $\pos$ in the proof of Corollary~\ref{cor:pairs}, as the sum over $j, j'$ of $\bbP_{\pos}(\refidx_{t+1} = j, \outidx_{t+1} = j')$ times~\eqref{eq:pairs}.

\subsection{The final forced move}\label{sec:formalization-washout}

In the proof of Proposition~\ref{prop:washout}, the bound by the variance of $\psi_S$ under $\weightvec_T$ is derived without the Cauchy--Schwarz inequality.
By the reversibility of $\forced_{\weightvec_T}$ with respect to $\weightvec_T$,
\begin{align*}
  & \langle \psi_S, \forced_{\weightvec_T} \psi_S \rangle_{\weightvec_T} = \weightvec_T(\psi_S^2) \\
  & \quad - \frac{1}{2} \sum_{j, j'} \weight_T^j \forced_{\weightvec_T}(j, j') \{\psi_S(j) - \psi_S(j')\}^2.
\end{align*}
Because $\forced_{\weightvec_T}$ leaves $\weightvec_T$ invariant, the bound
\begin{align*}
  & \{\psi_S(j) - \psi_S(j')\}^2 \leq 2 \{\psi_S(j) - \weightvec_T(\psi_S)\}^2 \\
  & \quad + 2 \{\psi_S(j') - \weightvec_T(\psi_S)\}^2
\end{align*}
shows that the last term is at most twice the variance, so that $\weightvec_T(\psi_S)^2 - \langle \psi_S, \forced_{\weightvec_T} \psi_S \rangle_{\weightvec_T}$ is at most the variance.

\subsection{Multinomial resampling and potential functions of the current state}\label{sec:formalization-reduction}

In the proof of Proposition~\ref{prop:reduction}, the marginal densities of $(\pos, \anc)$ under the two algorithms are compared directly.
Steps~\ref{alg:previous-state:initial} and~\ref{alg:previous-state:forward} of Algorithm~\ref{alg:previous-state} with the identity kernels are not used.
The sum of~\eqref{eq:extended-previous} with the Gibbs kernels over $\refidx_{1:T}$ is the function of $(\pos, \anc)$ in the factorization of~\eqref{eq:extended-previous} after Lemma~\ref{lem:extended-previous}, because the other factor is a probability distribution of $\refidx_{1:T}$.
With multinomial resampling and potential functions of the current state, this function is the sum of~\eqref{eq:extended} over $\refidx_T$, which is the marginal density of $(\pos, \anc)$ under Algorithm~\ref{alg:csmc}, by Lemma~\ref{lem:extended}.

\subsection{Assumptions and conventions of the formalization}\label{sec:formalization-assumptions}

The kernels of Algorithms~\ref{alg:csmc} and~\ref{alg:previous-state} are assumed to be measurable functions of the particle positions and, for Algorithm~\ref{alg:previous-state}, of the ancestor indices.
The kernels of the six selection rules are proved to satisfy this assumption.
The resampling schemes are assumed to be measurable functions of the particles and of their weights, and to satisfy the condition of Definition~\ref{def:resampling} for every probability vector of weights.
The results on the update rates, in Theorem~\ref{thm:ordering}, Corollary~\ref{cor:pairs}, and Propositions~\ref{prop:washout}, \ref{prop:gain}, and~\ref{prop:orderings}, assume that the diagonal of $\statespace^2$ is measurable, so that $\{\tilde{x}_t \neq x_t\}$ is an event.
Where the denominator of~\eqref{eq:bs} or~\eqref{eq:bs-previous} vanishes, $\bsprob_t(\cdot \mid j)$ is defined as $\weightvec_t$.

\subsection{Claims which are not formalized}\label{sec:formalization-claims}

All claims below are minor and affect the comparison or interpretation of the results, but not the nature of the main construction of the paper.
The experiments of Section~\ref{sec:numerics} and of Sections~\ref{sec:locally-optimal} and~\ref{sec:stratified} are not formalized, nor are the computations for the locally optimal proposal in Section~\ref{sec:locally-optimal}.
Neither are some comparison remarks with other methods in Sections~\ref{sec:intro} and~\ref{sec:discussion}, such as the effect of the steps of \citet{bunch2013improved} on the update rates.
The claims on fixing the indices of the reference, in Section~\ref{sec:method} and in Sections~\ref{sec:previous-state} and~\ref{sec:locally-optimal}, are not formalized as they are common knowledge in the field (and our proofs do not use this fact).

In Section~\ref{sec:validity}, the description of the backward pass of \BS{} as a Gibbs update of $(\refidx_T, \anc)$ given $\pos$ is not formalized.
In Section~\ref{sec:update-rates}, the remarks after Theorem~\ref{thm:ordering} and Corollary~\ref{cor:pairs} are not formalized.
These are the orderings of the lag-one autocorrelations and of the ESJD, the equality of the distributions of $(x_T, \tilde{x}_T)$ under \FMBS{} and \FMAT{} and under \BS{} and \AT{}, the pairwise distinctness of the particles when $\statespace = \mathbb{R}^d$ and $\dd x$ is Lebesgue measure, which Section~\ref{sec:numerics} also uses, and the reversal of the bound of Lemma~\ref{lem:coupling} for $\varphi(\cdot, j') = -\varphi(\cdot, j)$.
In Section~\ref{sec:gains}, the consequence of Proposition~\ref{prop:washout} with $\mixconst = 1$ for models which factorize across time is not formalized, nor is the interpretation of $\gain_t(j, j')$ as the decrease of a conditional probability.
In Section~\ref{sec:imh}, the probabilities $\bbP_{\pos}(\outidx_t = \refidx_t)$ with uniform probabilities $\weightvec_T$ and $\bsprob_t(\cdot \mid j)$ are formalized as probabilities under the backward Markov chain of Lemma~\ref{lem:pair}.
Their limit as $T - t$ grows is not formalized, nor are the interpretation of Proposition~\ref{prop:imh} as the independent Metropolis--Hastings sampler and the proportion of the time steps at which the moves of \MBS{} follow the Metropolis rule.

In Section~\ref{sec:previous-state}, the almost sure positivity of the denominator of~\eqref{eq:bs-previous} is not formalized.
The formalized results do not use it, because of the convention of Section~\ref{sec:formalization-assumptions}.
Only multinomial resampling is proved to satisfy Definition~\ref{def:resampling}, whereas stratified, systematic, residual, and killing resampling are taken to satisfy it from \citet[Section~3]{finke2026resampling}.
The failure of the second part of Lemma~\ref{lem:extended} in this setting, of which Section~\ref{sec:locally-optimal-ancestors} gives an example, is not formalized.
Neither is the distribution of $(\refidx_{t+1}, \anc_t)$ in \algstep{alg:previous-state:forward} with multinomial resampling.
Finally, neither are the reductions, with multinomial resampling and potential functions of the current state, of~\eqref{eq:extended-previous} with the identity kernels to~\eqref{eq:extended}, and of Algorithm~\ref{alg:ancestors} to the usual implementation of Section~\ref{sec:method}.
\end{appendix}

\end{document}